\documentclass[a4paper,twoside,11pt]{article}

\usepackage{BeaCed_article}

\title{Critical Ising model and spanning trees partition functions.}
\author{B\'eatrice de Tili\`ere
\thanks{{\small 
UPMC Univ Paris 06, UMR 7599,
Laboratoire de Probabilit\'es et Mod\`eles Al\'eatoires, 4 place Jussieu, 
F-75005 Paris.}
{\small\texttt{beatrice.de\_tiliere@upmc.fr.}}
{\small Supported by the ANR Grant 2010-BLAN-0123-02.}
}}
\date{}

\begin{document}

\maketitle

\begin{abstract}
We prove that the squared partition function of the two-dimensional critical Ising model defined on a finite, isoradial graph 
$\Gs=(\Vs,\Es)$, is equal to $2^{|\Vs|}$ times the partition function of spanning trees of the graph 
$\Gsext$, where $\Gsext$ is the graph $\Gs$ extended along the boundary; edges of $\Gs$ are assigned 
Kenyon's \cite{Kenyon3}  critical weights, and boundary edges of $\Gsext$ have specific weights. The proof is an explicit construction, providing a new relation 
on the level of configurations between two classical, critical models of statistical mechanics.
\end{abstract}

\section{Introduction}

Let $\Gs=(\Vs,\Es)$ be a finite, planar graph satisfying a geometric 
condition called \emph{isoradiality}, then every edge $e$ of $\Gs$ is naturally assigned an angle $\theta_e$. 
Consider the Ising model defined on the isoradial graph $\Gs$,
with Baxter's \cite{Baxter:Zinv} critical coupling constants $\Js=(\Js_e)_{e\in\Es}$:
\begin{equation*}
\forall\,e\in\Es,\quad \Js_e=\frac{1}{2}\log\left(\frac{1+\sin\theta_e}{\cos\theta_e}\right).
\end{equation*}
The partition function of the Ising model is denoted by $\Zising(\Gs,\Js)$, detailed definitions are given in Sections 
\ref{sec:isoradialgraphs} and \ref{sec:criticalversions}. Consider the graph $\Gsext$ obtained from the graph $\Gs$ by extending 
it along the boundary as in Section \ref{sec:ExtendedDefinitions}. The main result of this paper can loosely be stated as follows, a 
precise statement is given in Theorem~\ref{thm:MainResult} of Section \ref{sec:3}.
\begin{thm}\label{thm:main1}

The squared partition function of the critical Ising model defined on the isoradial graph $\Gs$,
is equal to $2^{|\Vs|}$ times the spanning trees partition function of the extended graph~$\Gsext$, 
with Kenyon's \cite{Kenyon3} critical weights on edges of $\Gs$, and specific weights on boundary
edges of $\Gsext$: 
\begin{equation*}
\Zising(\Gs,\Js)^2=
2^{|\Vs|}\Biggl|\sum_{\Ts\in \T^{\rs,\ss}(\Gsext)}
\Bigl(\prod_{e=(x,y)\in\Ts\cap\Es}\tan\theta_e\Bigr)
\underbrace{
\Bigl(\prod_{(x,y)\in\Ts\cap(\Esext\setminus\Es)}\tau_{(x,y)}\Bigr)
\Bigl(\prod_{(u,v)\in\Ts^*\cap(\Edext\setminus{\Es}^*)}\tau_{(u,v)}\Bigr)}_{\text{boundary contributions}}\Biggr|.
\end{equation*}
The sum is over spanning trees of the extended graph $\Gsext$, rooted at a vertex $\rs$, such that dual
spanning trees are rooted at a vertex $\ss$ of the dual $\Gdext$, see Section \ref{sec:MainResult}.
Weights $\tau$ at the boundary are defined in Equation \eqref{equ:weighttau}.
\end{thm}

\begin{rem}
A version of Theorem \ref{thm:main1} in the case where the graph $\Gs$ is embedded in the torus rather than the plane, is proved 
in the papers
\cite{BoutillierdeTiliere:iso_perio}, \cite{deTiliere:mapping} and \cite{Cimasoni:Critical}. The first two papers start from 
Fisher's correspondence \cite{Fisher} between the
Ising model defined on a graph~$\Gs$ and the dimer model defined on a non-bipartite, decorated graph $\Gs^{\mathrm F}$.
The squared Ising partition function is replaced by the determinant
of the dimer characteristic polynomial, and spanning trees are replaced by cycle rooted spanning forests of the torus.
In the paper \cite{deTiliere:mapping}, we prove an explicit mapping from cycle rooted spanning forests of the torus, to 
configurations counted by the determinant of the 
dimer characteristic polynomial. The paper \cite{Cimasoni:Critical} uses Kac-Ward matrices \cite{KacWard} instead of 
Fisher's correspondence. Working with graphs embedded in the torus has
the difficulty of involving the geometry of the torus, but avoids handling problems related to the graph having a boundary. 
\end{rem}

In this paper, we address issues related to the boundary of the graph, requiring a number of technical steps. The approach 
used to prove Theorem \ref{thm:main1} is also very different: it neither uses Fisher's correspondence nor Kac-Ward matrices.
Rather, our starting point is a mapping from the double Ising model defined on a planar graph $\Gs$, consisting of two independent Ising models,
to a dimer model defined on a related bipartite graph $\GQ$. This is a combination of a result of \cite{Nienhuis} and of
\cite{WuLin}-\cite{Dubedat}: Nienhuis \cite{Nienhuis} proves a mapping from the low temperature expansion of the two Ising models, to the 
free-fermion 6-vertex model; Wu and Lin \cite{WuLin} prove a mapping from the free-fermion 6-vertex model, to the bipartite dimer model in the case of the square lattice,
Dub\'edat \cite{Dubedat} generalizes this mapping to other planar graphs. Note that a generalization, to 
graphs embedded in surfaces of genus $g$, of the mapping from
the double Ising model to the bipartite dimer model is proved in 
\cite{BoutillierdeTiliere:XORloops}. 

Note also, that using a different approach, 
Dub\'edat \cite{Dubedat} relates two independent Ising models, with one living on the graph $\Gs$
and the second on the dual graph $\Gs^*$, to the bipartite dimer model on the graph~$\GQ$.
He uses a sequence of mappings, most of which are present in the physics literature: mapping from the double Ising model
to the 8-vertex model \cite{KadanoffWegner,Wu71}, then mapping from the 8-vertex model to another 8-vertex model
using Fan and Wu's abelian duality \cite{FanWu}; when coupling constants of the two Ising models satisfy Kramers and Wannier's
duality \cite{KramersWannier1,KramersWannier2}, the second 8-vertex model is a free-fermion
6-vertex model. The mapping from the free-fermion 6-vertex model to the bipartite dimer model is the one mentioned above.

The proof of Theorem \ref{thm:main1} provides an explicit mapping between the critical dimer model
on the bipartite graph $\GQ$ and critical spanning trees of the extended graph~$\Gsext$.

\medskip

\textbf{Outline of the paper}
\begin{itemize}
 \item \underline{Section \ref{sec:2}}. Definition of the Ising model on a planar graph $\Gs$. Definition of 
the dimer model on the bipartite graph $\GQ$ related to the double Ising model. Statement of the result 
\cite{Nienhuis,WuLin,Dubedat,BoutillierdeTiliere:XORloops} proving that 
the squared Ising partition function is equal, up to a constant,
to the partition function of the dimer model on the bipartite graph $\GQ$. Definition of the critical versions of the models on 
isoradial graphs. 
 \item \underline{Section \ref{sec:3}}. 
 Definition of spanning trees and related notions. Definition of the extended 
 graph $\Gsext$ constructed from $\Gs$. Statement of the main result, Theorem \ref{thm:MainResult}, loosely
 stated in Theorem \ref{thm:main1}.
\end{itemize}
Sections \ref{sec:4}, \ref{sec:5} and \ref{sec:6} consist of the proof of Theorem
\ref{thm:MainResult}. It provides a mapping between the critical dimer model on the bipartite graph $\GQ$ and
critical spanning trees of the extended graph $\Gsext$.
\begin{itemize}
 \item \underline{Section \ref{sec:4}}. Computation of the dimer partition function of the graph $\GQ$ 
 using a refinement of Kasteleyn/Temperley-Fisher's method \cite{Kasteleyn2,TF}, due to Kuperberg 
 \cite{Kuperberg}: it is equal to the determinant of a modified Kasteleyn matrix $\Ks$.
 \item \underline{Section \ref{sec:5}}. Interpretation of the Kasteleyn matrix $\Ks$ as the Laplacian matrix of 
 a directed graph $\GOo$ constructed from the bipartite graph $\GQ$, implying that
 the determinant of $\Ks$ counts weighted, oriented spanning trees of $\GOo$. 
 Transformation of the directed graph $\GOo$ along the boundary
 into a directed graph $\GO$: \emph{this is the first technical step used to handle the boundary of the 
 graph.}
 \item \underline{Section \ref{sec:6}}. Instead of considering oriented spanning trees of $\GO$,
 one considers dual spanning trees of the dual graph $\GOd$. Modification of the graph 
$\GOd$ along the boundary so that it becomes the \emph{double graph}~$\GDext$ of the extended graph~$\Gsext$
and of its dual~$\Gdext$: \emph{this is the second step used to handle the boundary of the graph.}
Characterization of spanning trees of the double graph~$\GDext$ arising as duals of oriented spanning trees 
of~$\GO$; mapping from perfect matchings of 
the double graph~$\GDext$ to spanning
trees of~$\GDext$: \emph{these are the two key steps allowing to reach spanning trees of the extended graph $\Gsext$.}
Conclusion of the proof of Theorem \ref{thm:MainResult} using the generalized form of Temperley's bijection \cite{Temperley}
due to Kenyon, Propp and Wilson \cite{KPW}.
\end{itemize}

\section{Square of Ising partition function via bipartite dimers, critical case}\label{sec:2}

In this section, we define the two-dimensional Ising model, the bipartite dimer model related to the double Ising model and the 
critical versions of the models.

In the whole of the paper, we suppose that the graph $\Gs=(\Vs,\Es)$ is planar, finite, connected, simple, with
vertices of degree at least two.

\subsection{Two-dimensional Ising model}\label{sec:FreeBoundary}

Consider a finite, planar graph $\Gs=(\Vs,\Es)$, together with a collection of positive \emph{coupling constants} 
$\Js=(\Js_e)_{e\in \Es}$ indexed by edges of $\Gs$. The \emph{Ising model on $\Gs$, with coupling constants~$\Js$}, is defined as follows.
A \emph{spin configuration} $\sigma$ is a function of the vertices of $\Gs$ taking values in $\{-1,1\}$. The probability 
of occurrence of a spin configuration $\sigma$ is given by the \emph{Ising Boltzmann
measure} $\PPising$, defined by:
$$
\PPising(\sigma)=\frac{1}{\Zising(\Gs,\Js)}\exp\Bigl(\sum_{e=uv\in
\Es}\Js_e\sigma_u\sigma_v\Bigr),
$$
where $\Zising(\Gs,\Js)=\sum\limits_{\sigma\in\{-1,1\}^\Vs}
\exp\Bigl(\sum\limits_{e=uv\in \Es}\Js_e\sigma_u\sigma_v\Bigr)$
is the normalizing constant, known as the \emph{Ising partition function}.

\subsection{Boundary conditions}

Suppose that the graph $\Gs$ is embedded in the plane, and also denote by $\Gs$ the embedded version of the graph. 
The set of \emph{boundary vertices}, respectively \emph{boundary edges}, of the graph $\Gs$, consists of vertices, respectively
edges, of the boundary of the outer face of $\Gs$. Let us denote by $\partial\Vs$ the set of boundary vertices of $\Gs$.

The Ising model defined in Section \ref{sec:FreeBoundary} is also known as the \emph{Ising model with free-boundary conditions}.
We now define the Ising model on $\Gs$, with coupling constants $\Js$ and \emph{plus-boundary conditions}.
A \emph{spin configuration} $\sigma$ is a function of the vertices of $\Gs$ taking values in $\{-1,1\}$
with the additional constraint that the value is $1$ on boundary vertices. Let us denote by 
$\{-1,1\}_{(\partial\Vs,+)}^{\Vs}$ this set of spin configurations. The probability of occurrence of a spin configuration $\sigma$ is defined similarly to 
the free-boundary case. The normalizing constant is denoted by $\Zising^+(\Gs,\Js)$, and referred to as the \emph{Ising partition function with
plus-boundary conditions}. It is equal to:
$$
\Zising^{+}(\Gs,\Js)=\sum\limits_{\sigma\in\{-1,1\}_{(\partial\Vs,+)}^{\Vs}}
\exp\Bigl(\sum\limits_{e=uv\in \Es}\Js_e\sigma_u\sigma_v\Bigr).
$$

The Ising partition function with plus-boundary conditions can be written using the
Ising partition function with free-boundary conditions of a related graph, in the following way.
Let us denote by $\Es_{\scriptscriptstyle \mathrm{\partial\Vs}}$ the set of edges of $\Gs$ joining boundary vertices.
Note that an edge of $\Es_{\scriptscriptstyle \mathrm{\partial\Vs}}$ is not necessarily
a boundary edge of the graph $\Gs$.


Define $\Gs'=(\Vs',\Es')$ to be the graph obtained from $\Gs$ by merging all
boundary vertices and edges of $\Es_{\scriptscriptstyle \mathrm{\partial\Vs}}$ into a single vertex $u_0$. 
Then, the graph $\Gs'$ is again planar. Observing that $\Es\setminus\Es_{\scriptscriptstyle \mathrm{\partial\Vs}}=\Es'$, we have:


\begin{align*}
\Zising^{+}(\Gs,\Js)&=\Bigl(\prod_{e\in\Es_{\scriptscriptstyle \mathrm{\partial\Vs}}}e^{\Js_e}\Bigr)
\sum\limits_{\sigma\in\{-1,1\}^{\Vs}_{(\partial\Vs,+)}}
\Bigl(\prod_{e=uv\in \Es\setminus\Es_{\scriptscriptstyle \mathrm{\partial\Vs}}}e^{\Js_e\sigma_u\sigma_v}\Bigr)\\
&=\Bigl(\prod_{e\in\Es_{\scriptscriptstyle \mathrm{\partial\Vs}}}e^{\Js_e}\Bigr)
\sum\limits_{\{\sigma\in\{-1,1\}^{\Vs'} :\,\sigma_{u_0}=1\}}
\Bigl(\prod_{e=uv\in \Es'}e^{\Js_e\sigma_u\sigma_v}\Bigr)\\
&=\frac{1}{2}\Bigl(\prod_{e\in\Es_{\scriptscriptstyle \mathrm{\partial\Vs}}}e^{\Js_e}\Bigr)
\sum\limits_{\sigma\in\{-1,1\}^{\Vs'}}
\Bigl(\prod_{e=uv\in \Es'}e^{\Js_e\sigma_u\sigma_v}\Bigr)\\
&=\frac{1}{2}\Bigl(\prod_{e\in\Es_{\scriptscriptstyle \mathrm{\partial\Vs}}}e^{\Js_e}\Bigr) \Zising(\Gs',\Js),
\end{align*}
where in the penultimate line, we have used the fact that the contribution $\prod_{e=uv\in
\Es'}e^{\Js_e\sigma_u\sigma_v}$ is invariant under the transformation
$\sigma\leftrightarrow-\sigma$.

As a consequence, the Ising partition function with plus-boundary conditions is
equal, up to an explicit multiplicative constant, to the partition function of the Ising
model with free-boundary conditions on a related, planar graph. The same kind of argument 
can be used to relate other boundary conditions to free ones.

In the sequel, we thus restrict ourselves to the Ising model with free-boundary conditions, simply referred to as the Ising model.
The result we prove holds for other boundary conditions as long as the related graphs satisfy the assumptions we shall make.

\subsection{Square of Ising partition function via bipartite dimers}

The goal of this section is to state Theorem \ref{thm:CedBea}, due to \cite{Nienhuis,WuLin,Dubedat,BoutillierdeTiliere:XORloops}, proving equality,
up to an explicit constant, between the partition function of the double Ising model on a planar graph $\Gs$ and the
partition function of the dimer model on a related bipartite graph $\Gs^\Qs$. The \emph{double Ising model} consists of two independent
Ising models living on the same graph $\Gs$, with the same coupling constants $\Js$, implying that its partition function is
the square of the partition function of the Ising model. 
Theorem \ref{thm:CedBea} of this section is actually the genus 0 case of the result of \cite{BoutillierdeTiliere:XORloops} which holds 
for graphs embedded in
surfaces of genus $g$. 

In the whole of this section, we suppose that the planar graph $\Gs$ is embedded in the sphere, and also denote by $\Gs$
the embedded graph. We start by defining the bipartite graph $\GQ$.

\subsubsection{Quadri-tiling graph}\label{sec:quadri}

Suppose that the embedding of the dual graph $\Gs^*$ of $\Gs$ is such that
primal and dual edges cross exactly once. The \emph{quad-graph}, denoted $\Gquad$,
is the graph whose vertices
are vertices of $\Gs$ and vertices of the dual graph $\Gs^*$. A dual vertex is
then joined to all primal vertices on the boundary of the corresponding
face. The embedding of $\Gquad$ is chosen such that its edges do not intersect
those of $\Gs$ and $\Gs^*$. Faces of the quad-graph are quadrangles and the
diagonals of the quadrangles consist of a primal edge of
$\Gs$ and its corresponding dual edge of $\Gs^*$, see Figure
\ref{fig:FigQuad} (left, dotted lines) for an example.

\begin{figure}[ht]
  \begin{center}
\includegraphics[width=13.5cm]{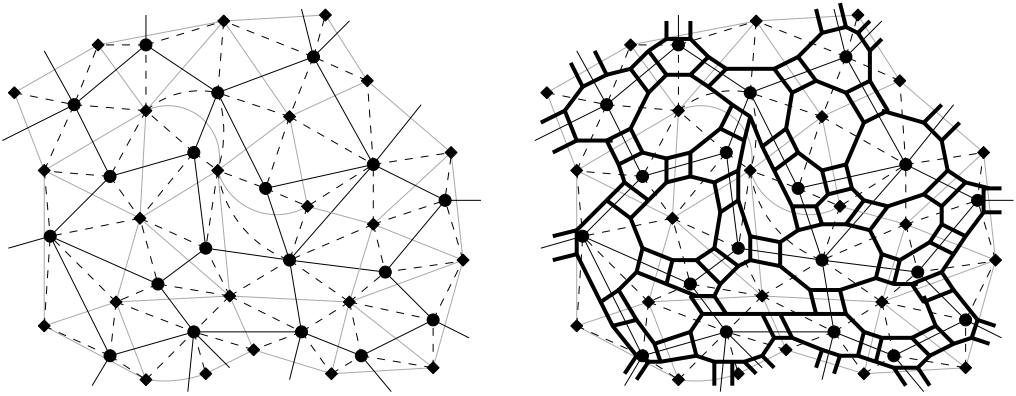}
    \caption{Left: a piece of a finite graph $\Gs$ embedded in the
sphere (plain full lines), of the dual graph $\Gs^*$ (grey full lines), of 
the quad-graph $\Gquad$ (dotted lines). Right: a piece of the quadri-tiling
graph $\GQ$ (thick full lines).}\label{fig:FigQuad}
  \end{center}
\end{figure}

Consider the
graph obtained by superimposing the primal graph $\Gs$, the dual graph $\Gs^*$,
the quad-graph $\Gquad$, and by adding a vertex at the crossing of each primal
and dual edge. Then, the dual of this graph, denoted by $\GQ=(\VQ,\EQ)$, is the
\emph{quadri-tiling} graph of $\Gs$, see Figure~\ref{fig:FigQuad}
(right, thick full lines) for an example. The graph $\GQ$ is bipartite, it contains: 
\begin{itemize}
\item \emph{quadrangles}, each quadrangle being contained in a 
quadrangle of the quad-graph $\Gquad$. In
each quadrangle, two edges cross a primal edge of $\Gs$ and two edges cross a dual edge of $\Gs^*$.
\item \emph{external edges}, crossing edges of the quad-graph $\Gquad$. 
\end{itemize}
Faces of the graph $\GQ$ are naturally partitioned as follows:
quadrangles corresponding to the crossing of primal and dual edges of $\Gs$ and $\Gs^*$,
faces corresponding to primal vertices of $\Gs$, and those corresponding to dual vertices of $\Gs^*$.
Note that all vertices of $\GQ$ have degree three.

\subsubsection{Dimer model and square of Ising partition function}

Consider an Ising model on the graph $\Gs$, with coupling constants $\Js=(\Js_e)_{e\in\Es}$. Using the coupling
constants $\Js$, define the following positive weight function $\nu=(\nu_\es)_{\es\in\EQ}$ on edges of the 
bipartite graph $\GQ$, see also Figure \ref{fig:FigQuad1}.

\begin{equation}\label{equ:dimerweights}
\nu_{\es}= 
\begin{cases}
1&\text{ if $\es$ is an external edge}\\
\cosh^{-1}(2\Js_{e})&\text{ if $\es$ crosses a primal edge $e$ of $\Gs$}\\
\tanh(2\Js_{e})&\text{ if $\es$ crosses the dual edge $e^*$ of an edge $e$ of $\Gs$}.
\end{cases}
\end{equation}

\begin{figure}[ht]
  \begin{center}
\psfrag{1}[c][c]{\scriptsize $1$}
\psfrag{e}[c][c]{\scriptsize $e$}
\psfrag{es}[c][c]{\scriptsize $e^*$}
\psfrag{cosh}[c][c]{\scriptsize $\cosh^{-1}(2\Js_e)$}
\psfrag{tanh}[c][c]{\scriptsize $\tanh(2\Js_e)$}
\includegraphics[width=5cm]{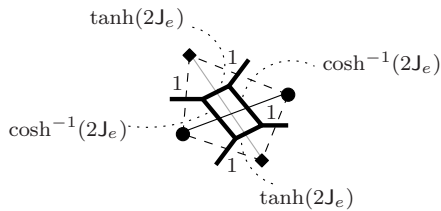}
    \caption{Positive weight function $\nu$ on edges of the bipartite graph $\GQ$.}\label{fig:FigQuad1}
  \end{center}
\end{figure}

The \emph{dimer model on $\GQ$ with weight function $\nu$}, is defined as follows.
A \emph{dimer configuration} of $\GQ$, also known as a \emph{perfect matching}, is a subset of edges of $\GQ$ such that
every vertex is incident to exactly one edge of the subset. Let us denote by
$\M(\GQ)$ the set of dimer configurations of the graph $\GQ$. The \emph{dimer model} assigns to a dimer configuration
$\Ms$ the probability of occurrence $\PPdimer(\Ms)$, equal to:

\begin{equation*}
\PPdimer(\Ms)= \frac{\prod_{\es\in M}\nu_{\es}}{\Zdimer(\GQ,\nu)}.
\end{equation*}
The probability measure $\PPdimer$ is known as the \emph{dimer
Boltzmann measure}; and the
normalizing constant $\Zdimer(\GQ,\nu)=\sum_{\Ms\in\M(\GQ)} \prod_{\es\in M}\nu_{\es}$, as the \emph{dimer partition function}.

The square of the Ising partition function and the 
dimer partition function are related by the following theorem.

\begin{thm}[\cite{Nienhuis,WuLin,Dubedat,BoutillierdeTiliere:XORloops}]\label{thm:CedBea}
Let $\Gs$ be a finite, planar graph embedded in the sphere, and let $\GQ$ be
the corresponding bipartite graph. Then,
\begin{equation*}
(\Zising(\Gs,\Js))^2=2^{|\Vs|}\bigl(\prod_{e\in
\Es}\cosh(2\Js_e)\bigr)\,\cdot\,\Zdimer(\GQ,\nu).
\end{equation*}
\end{thm}

\begin{rem}
Note that the graph $\GQ$ is defined using the spherical embedding of the
finite, planar graph $\Gs$, but once the graph $\GQ$ is defined, the dimer
partition function does not depend on the choice of embedding of $\GQ$.
Similarly, the Ising partition function does not depend on the embedding of the
graph $\Gs$.
\end{rem}

\subsection{Critical versions of the model}

In this section, we define the critical version of the Ising model defined on isoradial graphs, and the
critical version of the corresponding bipartite dimer model.

We need the following definition. If $\Gs$ is a planar graph, embedded in
the plane, the \emph{restricted dual graph} of $\Gs$, denoted $\Gdres$, is
the dual of the graph $\Gs$ from which are removed: the vertex corresponding to the outer face and the dual edges corresponding
to boundary edges of $\Gs$.

\subsubsection{Isoradial graphs}\label{sec:isoradialgraphs}

Isoradial graphs probably first appeared in the work of Duffin \cite{Duffin}, see also Mercat \cite{Mercat:ising}. 
The definition we present here and the name come from the paper \cite{Kenyon3} by Kenyon. A planar graph $\Gs$ is \emph{isoradial}, 
if it can be embedded in the plane
in such a way that all faces are inscribable in a circle, with all circles
having the same radius. The embedding is said to be \emph{regular}
if all circumcenters are in the closure of the faces. 

Let $\Gs$ be a finite, isoradial graph having a regular isoradial embedding. We fix such an embedding, take the
common radius to be 1, and also denote by $\Gs$ the embedded graph. A regular isoradial embedding of the 
restricted dual $\Gdres$, with radius 1, is obtained by taking as dual vertices
the circumcenters of the corresponding faces. An example is provided in Figure~\ref{fig:Iso0} (left: plain and grey lines).

\begin{figure}[ht]
\begin{center}
\psfrag{e}[c][c]{\scriptsize $e$}
\psfrag{te}[c][c]{\scriptsize $\theta_e$}
\includegraphics[width=13cm]{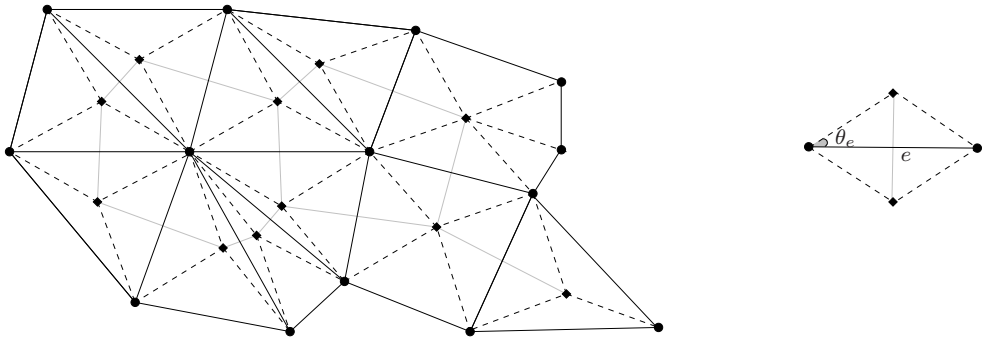}
\caption{Left: isoradial embedding of a finite planar graph $\Gs$ (plain black
lines) and of the restricted dual graph $\Gdres$ (grey lines); restricted
quad-graph $\Gquadres$ (dotted lines). Right: half-rhombus angles assigned to edges of $\Gs$.}
\label{fig:Iso0}
\end{center}
\end{figure}

The \emph{restricted quad-graph}, denoted $\Gquadres$, is the graph whose vertices are vertices of the graph~$\Gs$
and of the restricted dual $\Gdres$. A dual vertex of $\Gdres$ is joined to all primal
vertices on the boundary of the corresponding face, see Figure \ref{fig:Iso0} (left, dotted lines). Note that this definition does not
use isoradiality but, when the graphs $\Gs$ and $\Gdres$ are isoradially embedded, faces of the restricted quad-graph
are side-length 1 rhombi, or half-rhombi along the boundary of
the outer face. 

Every edge $e$ of $\Gs$ is the diagonal of exactly one
rhombus, or half-rhombus, of the restricted quad-graph, and we let $\theta_e$ be the half-angle at
the vertex it has in common with $e$, see Figure~\ref{fig:Iso0} (right).

\subsubsection{Critical versions of the Ising and dimer models}\label{sec:criticalversions}

Consider an Ising model defined on a finite, isoradial graph $\Gs$ with coupling
constants $\Js$, having a regular isoradial embedding. We fix such
an embedding and for every edge $e$ of $\Gs$, we  
let $\theta_e$ be the rhombus half-angle assigned to $e$.
The Ising model is said to be \emph{critical} if the coupling constants are equal to:
\begin{equation}\label{equ:IsingWeights}
\forall\,e\in\Es,\quad \Js_e=\frac{1}{2}\log\left(\frac{1+\sin\theta_e}{\cos\theta_e}\right). 
\end{equation}
These coupling constants were first derived by Baxter \cite{Baxter:Zinv}, using $Z$-invariance and a generalized
form of self-duality. They have been shown to be critical in the case of infinite periodic graphs by 
Li \cite{Li:critical}, and Cimasoni, Duminil-Copin \cite{CimasoniDuminil}.

Now consider a spherical embedding of the graph $\Gs$, and let $\GQ$ be the corresponding bipartite graph defined 
in Section \ref{sec:quadri}, on which the dimer model lives. Computing the dimer weights corresponding to 
the critical Ising weights, using Equation \eqref{equ:dimerweights} yields:

\begin{equation}\label{equ:KastWeights}
\forall\, \es\in\EQ,\quad \nu_\es=
\begin{cases}
1&\text{ if $\es$ is an external edge}\\
\cos(\theta_e)&\text{ if $\es$ crosses a primal edge $e$ of $\Gs$}\\
\sin(\theta_e)&\text{ if crosses the dual edge $e^*$ of an edge $e$ of $\Gs$}.
\end{cases}
\end{equation}

\section{Statement of main result}\label{sec:3}

The goal of this section is to state the main result of this paper proving that the square of the critical Ising model partition function is equal, up to
a multiplicative constant, to the partition function of critical spanning trees with appropriate boundary conditions. In 
Section~\ref{sec:LaplacianDirected}, we define spanning trees, the Laplacian matrix and state the matrix-tree theorem. Then, in
Section~\ref{sec:ExtendedDefinitions}, we consider the isoradial graph $\Gs$ on which the critical Ising model is defined,
extend the graph along its boundary into a graph $\Gsext$, and explain how the extended graph is embedded in the plane. 
Using this embedding we define a weight function on 
edges of $\Gsext$ and of its dual $\Gsext^*$, which turns out to be the spanning trees' critical weight function
of Kenyon \cite{Kenyon3}, with specific boundary conditions. 
In Section \ref{sec:MainResult}, we state Theorem \ref{thm:MainResult}, the main theorem of the paper.

\subsection{Spanning trees, Laplacian matrix, matrix-tree theorem}\label{sec:LaplacianDirected}

For the purpose of this section only, we let $\GV$ be a generic finite, directed graph; and we let
$\Gs$ be a generic undirected graph. Let $\rs$ denote a distinguished vertex of $\Gs$ or $\GV$, referred to as the 
\emph{root vertex}, or simply \emph{root}.

An \emph{$\rs$-rooted oriented spanning tree} ($^{\rs}\mathrm{OST}$) of $\GV$, is a subset of oriented edges of $\GV$ containing no cycle,
such that each vertex except the root $\rs$ has an outgoing edge. We denote by 
$\T^\rs(\GV)$ the set of $^{\rs}\mathrm{OSTs}$ of the graph $\GV$ rooted at $\rs$. Suppose that a weight function $\tau$ is assigned to oriented edges of $\GV$, meaning that an oriented edge 
$(\xs,\ys)$ has weight $\tau_{(\xs,\ys)}$. The weighted sum of $^{\rs}\mathrm{OSTs}$ of the graph 
$\GV$, is denoted $\ZOST^{\rs}(\GV,\tau)$ and referred to as the \emph{$^{\rs}\mathrm{OST}$ partition function}.
\begin{equation*}
\ZOST^{\rs}(\GV,\tau)=\sum_{\Ts\in\T^\rs(\GV)}\prod_{e=(\xs,\ys)\in\Ts}\tau_{(\xs,\ys)}.
\end{equation*}

The \emph{Laplacian
matrix} of the directed graph $\GV$, denoted by $\Delta_{\GV}$, has lines and
columns indexed by vertices of the graph. Coefficients of the matrix
$\Delta_{\GV}$ are defined by:
\begin{equation*}
(\Delta_{\GV})_{\xs,\ys}=
\begin{cases}
\tau_{(\xs,\ys)}&\text{ if $(\xs,\ys)$ is an oriented edge of $\GV$}\\
-\sum\limits_{\{\xs':\, (\xs,\xs')\text{ is an edge of $\GV$}\}}\tau_{(\xs,\xs')}&\text{ if $\ys=\xs$}\\
0&\text{ else}.
\end{cases}
\end{equation*}

The following theorem is the version for directed graphs of the classical theorem of Kirchhoff. It is due to
Tutte \cite{Tutte}, see also Chaiken \cite{Chaiken} for a nice statement.

\begin{thm}[\cite{Tutte}]\label{thm:matrixtree}
Let $\Delta_{\GV}$ be the Laplacian matrix of the directed graph $\GV$
with weight function $\tau$ on the edges. Let $\Delta_{\GV}^{(\rs)}$ be the
matrix $\Delta_{\GV}$ from which the line and the column corresponding to the root vertex $\rs$ have been removed. Then,
\begin{equation*}
\det(\Delta_{\GV}^{(\rs)})=\ZOST^{\rs}(\GV,\tau).
\end{equation*}
\end{thm}

Now consider an undirected graph $\Gs$. To the graph $\Gs$ is naturally associated a directed graph $\GV$, referred to
as the \emph{directed version} of $\Gs$: it has the same set of vertices, and two 
oriented edges $(\xs,\ys)$ and $(\ys,\xs)$ for every edge $\xs\ys$ of $\Gs$. An \emph{$\rs$-rooted oriented
spanning tree of $\Gs$}, is an $\rs$-rooted oriented spanning tree of its directed 
version $\GV$. Similarly to the directed case, $\T^\rs(\Gs)$ denotes the set of $^{\rs}\mathrm{OSTs}$ of the graph $\Gs$. A 
\emph{spanning tree} of the graph $\Gs$ is the unoriented version of an OST of $\Gs$. Note that a spanning tree
is the unoriented version of $|\Vs|$ OSTs, where $|\Vs|$ corresponds to the possible choices for the root vertex.

Suppose that a weight function $\tau$ is assigned to edges of $\Gs$, meaning that an edge $\xs\ys$ has
weight $\tau_{\xs\ys}$. This yields a symmetric weight function, also denoted $\tau$, on the directed version $\GV$, defined by: 
$\tau_{(\xs,\ys)}=\tau_{(\ys,\xs)}=\tau_{\xs\ys}$. The \emph{$^{\rs}\mathrm{OST}$ partition function} of the graph $\Gs$, 
denoted by $\ZOST^{\rs}(\Gs,\tau)$ is the $^{\rs}\mathrm{OST}$ partition function of the directed version $\GV$ of $\Gs$,
with the symmetric weight function induced by the weight function $\tau$ of the graph $\Gs$. 
Note that since the weight function is symmetric, it is equal to the weighted sum of spanning trees of $\Gs$. It does not
depend on the choice of root vertex.

The Laplacian matrix of the graph
$\Gs$, denoted $\Delta_\Gs$, is the Laplacian matrix of its directed version $\GV$; it is symmetric. In this case, 
Theorem~\ref{thm:matrixtree} is the classical version of the matrix-tree theorem \cite{Kirchhoff}.

A classical fact about spanning trees is that if $\Ts$ is a spanning tree of a graph $\Gs$ embedded in the plane and if 
$\Ts^*$ denotes the complement of the edges dual to $\Ts$, then $\Ts^*$ is a spanning tree of the dual graph $\Gs^*$, 
it is referred to as the \emph{dual spanning tree} of a spanning tree of~$\Gs$. If a spanning tree is oriented, by 
dual spanning tree, we mean the dual of its unoriented version. 

If a graph is partially directed, the above definitions are easily adapted.

\subsection{Extended graphs, isoradial embeddings}\label{sec:ExtendedDefinitions}

Let $\Gs$ be a finite, embedded planar graph. A planar embedding of the dual graph $\Gs^*$ is obtained by assigning a dual
vertex to the outer face of $\Gs$. The \emph{extended dual graph}, denoted $\Gdext=(\Vdext,\Edext)$, is obtained from the dual graph $\Gs^*$ by 
splitting the vertex corresponding to
the outer face of $\Gs$ into $n$ vertices, where $n$ is the number of boundary edges of $\Gs$, and by joining these vertices in a circular
way, see Figure \ref{fig:Iso2a} (grey lines) for an example. The \emph{extended graph} $\Gsext=(\Vsext,\Esext)$ is the dual graph of 
the extended dual, see Figure \ref{fig:Iso2a} (plain and thick black lines). Let us denote by $\rs$ the vertex of $\Gsext$
corresponding to the outer face of the extended dual $\Gdext$.

Note that the graph $\Gs$, the extended dual $\Gdext$, and the extended graph $\Gsext$ can be embedded in such a way that 
they are planar, and such that primal and dual edges cross at most once.

\begin{figure}[ht]
\begin{center}
\psfrag{o}[c][c]{\scriptsize $\rs$}
\includegraphics[width=8cm]{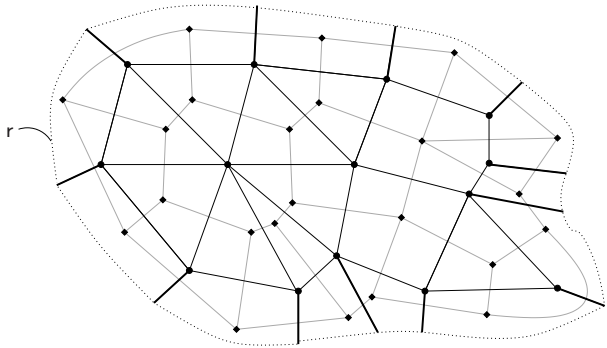}
\caption{A planar graph $\Gs$ (plain black lines), the extended dual $\Gdext$ (grey lines), and the extended graph $\Gsext$
(plain and thick black lines). The dotted line is a spread out way of representing the vertex $\rs$ corresponding to the 
outer face of $\Gdext$.}
\label{fig:Iso2a}
\end{center}
\end{figure}

The \emph{extended quad-graph}, denoted $\Gquadext$ is obtained from the restricted quad-graph $\Gquadres$ by adding the missing
half-quadrangles along the boundary of $\Gs$; in doing so, vertices of the extended dual graph are also added.

Suppose now that the graph $\Gs$ is finite, isoradial with a regular embedding, and fix such an embedding. Then, an embedding of
the extended quad-graph $\Gquadext$ is obtained by adding, in the natural 
embedding of $\Gquadres$, 
the missing half-rhombi along the boundary see Figure \ref{fig:Iso3} (dotted lines). 

Using this embedding, we now assign half-angles to edges of the extended graph $\Gsext$ that are incident to the vertex $\rs$. 
Consider an edge $e^\partial=x\rs$ of 
$\Gsext$ incident to the outer vertex $\rs$,
then it is the dual of an edge $(u,v)$ of the extended dual $\Gdext$, where $u$ comes before $v$ when going clockwise along
the boundary. The half-angle $\theta^\partial$ assigned to the edge $e^{\partial}$, is half the angle 
$\angle(xu,xv)$ defined by the embedding of the extended quad-graph $\Gquadext$, see Figure \ref{fig:Iso3}.

\begin{figure}[ht]
\begin{center}
\psfrag{x}[c][c]{\scriptsize $x$}
\psfrag{o}[c][c]{\scriptsize $\rs$}
\psfrag{u}[c][c]{\scriptsize $u$}
\psfrag{v}[c][c]{\scriptsize $v$}
\psfrag{e}[r][r]{\scriptsize $e^\partial$}
\psfrag{te}[c][c]{\scriptsize $\theta^\partial$}
\includegraphics[width=8cm]{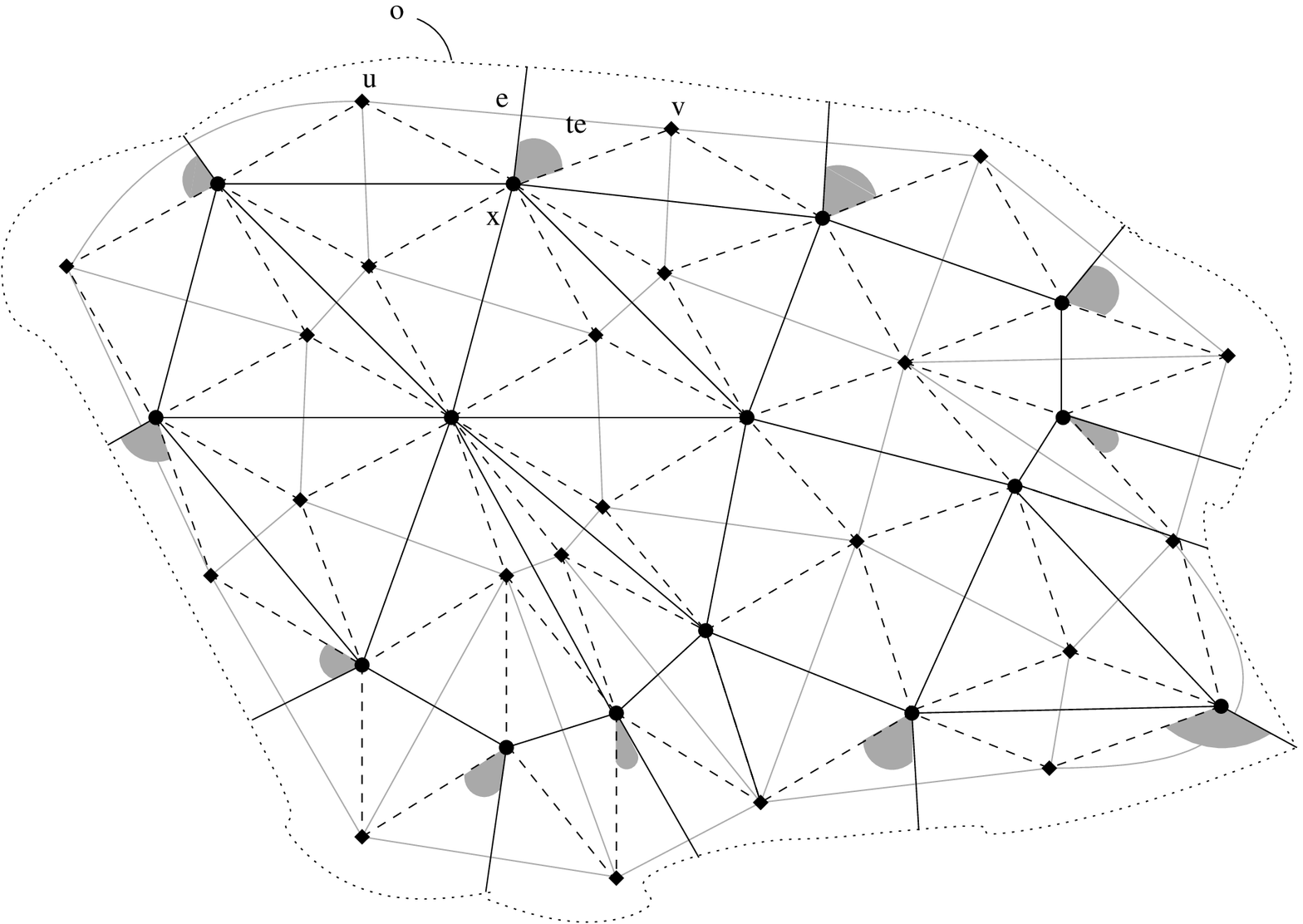}
\caption{Extended quad-graph $\Gquadext$ (dotted lines) and half-angles assigned to edges of $\Gsext$ incident to the outer 
vertex $\rs$.}
\label{fig:Iso3}
\end{center}
\end{figure}

\begin{rem}
Note that half-rhombi added to construct the extended quad-graph from the reduced quad-graph in the isoradial embedding,
may overlap. Then, half-angles assigned to edges along the boundary of the extended graph $\Gsext$ may
be negative; note that they may also be larger than $\frac \pi 2$.
\end{rem}

\subsection{Statement of main result}\label{sec:MainResult}

Let $\Gs$ be a finite, isoradial graph, with a fixed isoradial embedding. Consider the extended dual graph $\Gdext$ and the
extended graph $\Gsext$ defined in Section \ref{sec:ExtendedDefinitions}. We now define the weight function $\tau$
on edges of $\Gsext$ and $\Gdext$. It is defined on oriented edges, but is non-symmetric only on
the boundary of the graphs, implying that only boundary edges of the graphs are `really' oriented. If $uv$ is a boundary
edge of the extended dual graph $\Gdext$, then by convention when writing $(u,v)$, we mean that 
$u$ comes before $v$ when traveling clockwise along the boundary. The weight function $\tau$ is defined as follows.
\begin{align}\label{equ:weighttau}
&\begin{cases}
\tau_{(x,y)}=\tau_{(y,x)}=\tan\theta_e  &\text{ if $xy$ is an $e$ edge of $\Gs$}\\
\tau_{(x,\rs)}=2\sin\frac{\theta^{\partial}}{2}  &\text{ if $x\rs$ is an edge of $\Gsext$ incident to $\rs$}\\
\tau_{(\rs,x)}=0& 
\end{cases}\\\nonumber
&\begin{cases}
\tau_{(u,v)}=\tau_{(v,u)}=1  &\text{ if $uv$ is the dual of and edge $e$ of $\Gs$}\\
\tau_{(u,v)}=e^{-i\frac{\theta^\partial}{2}} & \text{ if $(u,v)$ is a boundary edge of $\Gdext$}\\
\tau_{(v,u)}=e^{i\frac{\theta^\partial}{2}}  &\text{ if $(u,v)$ is a boundary edge of $\Gdext$}.
\end{cases}
\end{align}
Note that away from the boundary, this is the spanning trees critical weight function introduced by Kenyon \cite{Kenyon3}.

Fix a root vertex $\ss$ on the boundary of the extended dual $\Gdext$. Given an oriented spanning tree of $\Gsext$ rooted 
at $\rs$, let us denote by $\Ts^*$ the dual spanning tree of $\Gdext$ rooted at the vertex $\ss$, and by
$\T^{\rs,\ss}(\Gsext)$ the set of $^{\rs}\mathrm{OST}$ $\Ts$ of $\Gsext$, such that the dual is rooted at $\ss$.

Let $\ZOST^{\rs,\ss}(\Gsext,\tau)$ be
the following weighted sum over spanning trees of $\Gsext$:
\begin{multline*}
\ZOST^{\rs,\ss}(\Gsext,\tau)
=\sum_{\Ts\in \T^{\rs,\ss}(\Gsext)}
\Bigl(\prod_{e=(x,y)\in\Ts\cap\Es}\tau_{(x,y)}\Bigr)
\Bigl(\prod_{(x,y)\in\Ts\cap(\Esext\setminus\Es)}\tau_{(x,y)}\Bigr)
\Bigl(\prod_{(u,v)\in\Ts^*\cap(\Edext\setminus{\Es}^*)}\tau_{(u,v)}\Bigr).
\end{multline*}

\begin{rem}\label{rem:KenyonWeights}
Note that dual spanning trees contribute only along the boundary of the extended dual graph $\Gdext$. 
If the boundary of the extended graph $\Gdext$ were undirected, we could factor out the contributions of 
dual spanning trees, this is not possible because 
the boundary of the 
extended graph $\Gdext$ is directed. 
\end{rem}

\begin{thm}\label{thm:MainResult}
Consider the critical Ising model defined on a finite isoradial graph $\Gs$, having a regular embedding, with critical
coupling constants $\Js$ assigned to edges. Let $\Gsext$ and $\Gdext$ be the extended versions of the graph $\Gs$ and of 
its dual graph, with weight function $\tau$ of Equation~\eqref{equ:weighttau} assigned to edges. Then,
\begin{equation*}
(\Zising(\Gs,\Js))^2=2^{|\Vs|}|\ZOST^{\rs,\ss}(\Gsext,\tau)|.
\end{equation*}
\end{thm}

The remainder of the paper consists in the proof of Theorem \ref{thm:MainResult}.
It starts from Theorem \ref{thm:CedBea} where the squared Ising partition function of the planar graph $\Gs$ is
shown to be equal to the dimer partition function of the bipartite graph $\GQ$.

\section{Kasteleyn theory for the critical dimer model on the graph~$\GQ$}\label{sec:4}

In this section we write an explicit formula for the partition function of the
dimer model on the bipartite graph $\GQ$ with weights $\nu$ assigned to edges, corresponding to the critical Ising weights.
We use the approach of Kuperberg~\cite{Kuperberg}, which generalizes the approach of Kasteleyn~\cite{Kasteleyn2}, and
Temperley and Fisher \cite{TF}. In Section \ref{sec:Kuperberg} we explain how Kuperberg's approach amounts to assigning appropriate 
phases to edges; then in Section \ref{sec:KastPhase} we specify a choice of phases for the dimer model on the graph $\GQ$ with critical weights. 
Our choice is
related to the one made by Kenyon in \cite{Kenyon3}.

\subsection{Explicit formula for the dimer partition function}\label{sec:Kuperberg}

Kasteleyn \cite{Kasteleyn2} and independently Temperley and Fisher \cite{TF} proved an
explicit formula for the partition function of the dimer model defined on a
finite, planar graph $\Gs$. When the graph $\Gs$ is moreover bipartite, this 
theorem has been specified by Percus \cite{Percus}. The set of
vertices $\Vs$ can naturally be split into two subsets $\Ws\cup\Bs$, where $\Ws$
denotes white vertices, $\Bs$ black ones, and vertices in $\Ws$ are only
incident to vertices in $\Bs$. Then, the dimer partition function is 
equal to the determinant of a weighted, oriented adjacency matrix of the graph
$\Gs$, whose lines are indexed by white vertices and columns by black ones.
Edges of the graph $\Gs$ are oriented according to an \emph{admissible
orientation}, that is an orientation of the edges such that all cycles
surrounding inner faces of the graph are \emph{clockwise odd}, \emph{i.e.} when
traveling clockwise around such a cycle, the number of co-oriented edges is odd.
Kasteleyn proved that such an orientation always exists \cite{Kasteleyn2}.

In our case, instead of an admissible orientation, it is more convenient to use
an appropriate \emph{phasing of the edges}: every edge
$e=wb$ is assigned a modulus-one, complex number $e^{i\phi_{wb}}$, where
$\phi_{wb}$ is referred to as the \emph{phase} of the edge $e$. Following
Kuperberg
\cite{Kuperberg}, the phasing has to satisfy the following condition. Let $F$ be an inner
face of the graph $\Gs$, whose boundary vertices are $w_1,b_1,\cdots,w_k,b_k$ in
clockwise order. The \emph{Kasteleyn curvature} at the face $F$, denoted
by
$C(F)$, is:
\begin{equation*}
C(F)=(-1)^{k-1}
\frac{\prod_{j=1}^k e^{i\phi_{w_j b_j}}}{\prod_{j=1}^k e^{i\phi_{w_{j+1} b_j}}}.
\end{equation*}
The phasing of the edges is \emph{flat} if, for every inner face $F$ of the
graph $\Gs$, $C(F)=1$. 

Suppose that edges of the graph $\Gs$ are assigned positive weights
$(\nu_e)_{e\in\Gs}$. Define the \emph{Kasteleyn matrix} $\Ks$ to be 
the following complex valued, adjacency matrix of the graph $\Gs$, having
lines indexed by white vertices and columns by black ones:
\begin{equation}\label{equ:KastMatrix}
\forall\, w\in\Ws,\,b\in\Bs,\quad \Ks_{w,b}=
\begin{cases}
e^{i\phi_{wb}}\nu_e&\text{ if $e=wb$ is an edge of $\Gs$}\\
0&\text{ else}.
\end{cases}
\end{equation}
Then Kuperberg proves the following variant of the result of Kasteleyn \cite{Kasteleyn2}:
\begin{thm}\emph{\cite{Kuperberg}}\label{thm:Kuperberg}
If the phasing of the edges is flat, the partition function of the dimer model
on $\Gs$ is equal to:
$$
\Zdimer(\Gs,\nu)=|\det(\Ks)|.
$$
\end{thm}

\subsection{Choice of phases}\label{sec:KastPhase}

Let $\Gs$ be a finite, planar, embedded graph.
Recall that the bipartite graph $\GQ$ is constructed from a spherical embedding of the graph $\Gs$ and of its dual graph $\Gs^*$.
A planar embedding of $\GQ$ is obtained by taking a dual vertex of $\Gs^*$ as the outer face of $\GQ$. 
Figure \ref{fig:Iso1} (thick black lines) is a picture of a planar embedding of the graph $\GQ$. 
Note that grey lines of Figure \ref{fig:Iso1} 
represent the extended dual $\Gdext$ rather than
the dual $\Gs^*$, so that the outer face of the graph $\GQ$ contains all dual vertices arising from the splitting of the dual
vertex corresponding to the outer face.

\begin{figure}[ht]
\begin{center}
\psfrag{o}[c][c]{\scriptsize $\rs$}
\includegraphics[width=10cm]{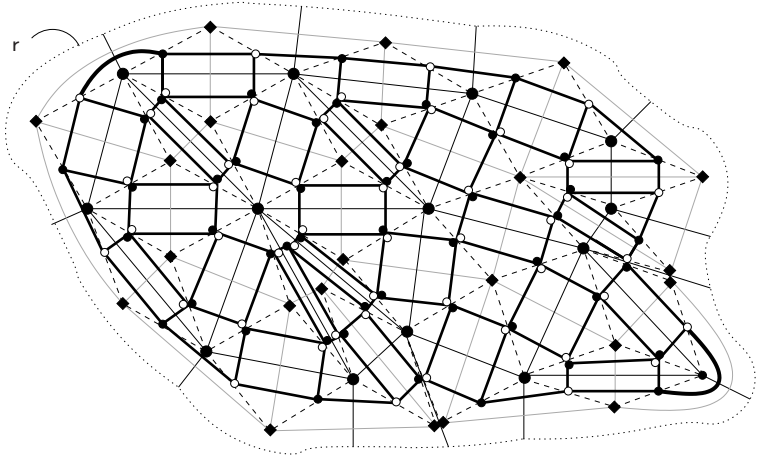}
\caption{Near isoradial embedding of the graph $\GQ$.}\label{fig:Iso1}
\end{center}
\end{figure}

Suppose moreover that the graph $\Gs$ is isoradial with a regular isoradial embedding, and fix such an embedding. Consider also
the induced embedding of the extended quad-graph $\Gquadext$, see Figure \ref{fig:Iso1} (dotted lines). Using this,
we now define a near-isoradial embedding of the graph $\GQ$, which we then use to define phases assigned to edges of $\GQ$.
Draw quadrangles of the bipartite graph $\GQ$ in such a way that they are rectangles included in the rhombi
of the extended quad-graph; draw vertices of the rectangles as midpoints of the sides of the rhombi. An external edge 
which is not on the boundary of the outer face of $\GQ$ has length 0. An external edge on the boundary of the outer face is drawn
so that the face contains the primal vertex in its interior.

\begin{rem}$\,$
\begin{itemize}
\item Away from the boundary, this defines a regular isoradial embedding of the graph $\GQ$ with radius $\frac 1 2$. Note that
the critical dimer weights on the graph $\GQ$ induced by the critical Ising model, are the \emph{critical}
dimer weights defined by Kenyon \cite{Kenyon3} for bipartite isoradial graphs. 

\item Along the boundary, if one draws straight lines to join boundary vertices, then boundary faces are inscribed in a circle
of radius $\frac 1 2$, but it might be that the circumcenter is outside of the face (this happens in the lower right corner of
Figure \ref{fig:Iso1}). The embedding is thus isoradial, but not regular. What can also happen is that 
rhombi of the extended quad-graph $\Gquadext$ overlap. In this case, the embedding is not planar. One can nevertheless define 
the notion of face, and show that it can be embedded in a circle of radius $\frac 1 2$.

\item We use this embedding only to define phases assigned to edges of $\GQ$. 
\end{itemize}
\end{rem}

The bipartite coloring of the vertices
is fixed so that at each rectangle, it is as in
Figure~\ref{fig:FigRhombusColoring}.

\begin{figure}[ht]
\begin{center}
\includegraphics[height=2cm]{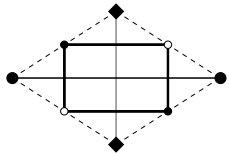}
\caption{Bipartite coloring of the vertices of the bipartite graph $\GQ$.}\label{fig:FigRhombusColoring}
\end{center}
\end{figure}

We now define phases assigned to edges of $\GQ$. Recall that edges of $\GQ$ either belong to quadrangles/rectangles, or are external edges.
\begin{enumerate}
\item Each rectangle belongs to a rhombus of
the extended quad-graph $\Gquadext$, which contains an edge $e$ of $\Gs$. An edge $\es=\ws\bs$ of the rectangle either crosses the
dual edge $e^*$ of the edge $e$, in which case we let $\phi_{\ws\bs}=0$; or it crosses the edge $e$, and 
we let $\phi_{\ws\bs}=\frac{\pi}{2}$.
\item[2(a)] An external edge $\es=\ws\bs$ which is not on the boundary of $\GQ$ crosses an edge shared by two rhombi of the extended
quad-graph. The vertex $\ws$ belongs to one of the two rhombi, and
we let $\phi_{\ws\bs}=\frac{3\pi}{2}-\theta_{e(\ws)}$, where
$\theta_{e(\ws)}$ is the rhombus half-angle of the edge $e(\ws)$ of $\Gs$
contained in the rhombus, see Figure~\ref{fig:Phase1}.

\begin{figure}[ht]
\begin{center}\psfrag{e1}[c][c]{\scriptsize $e_1$}
\psfrag{e}[c][c]{\scriptsize $e(\ws)$}
\psfrag{te}[c][c]{\scriptsize $\theta_{e(\ws)}$}
\psfrag{tb}[c][c]{\tiny $\theta_{\partial}$}
\psfrag{w}[c][c]{\scriptsize $\ws$}
\psfrag{b}[c][c]{\scriptsize $\bs$}
\includegraphics[height=2cm]{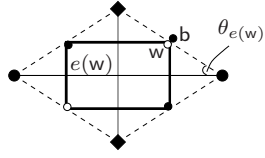}
\caption{Definition of the phase
$\phi_{\ws\bs}$ of an external edge $\ws\bs$ not on the boundary of $\GQ$.}\label{fig:Phase1}
\end{center}
\end{figure}

\item[2(b)] An external edge $\es=\ws\bs$ on the boundary of $\GQ$ crosses an edge $e^\partial$ of the extended graph $\Gsext$. 
Moreover, the vertex $\ws$
belongs to a quadrangle contained in a rhombus of the extended quad-graph.
We let $\phi_{\ws\bs}=\frac{3\pi}{2}-(\theta_{e(\ws)}+\theta^\partial)$, where
$\theta_{e(\ws)}$ is the rhombus half-angle of the edge $e(\ws)$ of $\Gs$
contained in the rhombus, and $\theta^\partial$ is the half-angle assigned in Section \ref{sec:ExtendedDefinitions} 
to edges of the extended graph $\Gsext$, incident to the vertex $\rs$, see also Figure \ref{fig:Phase2}.

\begin{figure}[ht]
\begin{center}\psfrag{e1}[c][c]{\scriptsize $e_1$}
\psfrag{e}[c][c]{\scriptsize $e(\ws)$}
\psfrag{eb}[c][c]{\scriptsize $e^\partial$}
\psfrag{te}[c][c]{\scriptsize $\theta_{e(\ws)}$}
\psfrag{tb}[c][c]{\scriptsize $\theta^\partial$}
\psfrag{w}[c][c]{\scriptsize $\ws$}
\psfrag{b}[c][c]{\scriptsize $\bs$}
\psfrag{o}[c][c]{\scriptsize $\rs$}
\includegraphics[height=2.5cm]{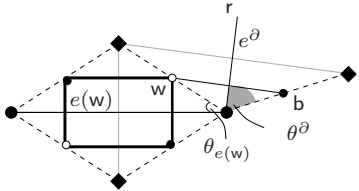}
\caption{Definition of the phase
$\phi_{\ws\bs}$ of an external edge $\ws\bs$ on the boundary of $\GQ$.}\label{fig:Phase2}
\end{center}
\end{figure}

When no confusion occurs, we shall omit the argument $\ws$ in the notation of
$\theta_{\es(\ws)}$.
\end{enumerate}

\begin{lem}\label{lem:phasing}
The phasing of the edges of $\GQ$ defined by Points 1. and 2. above is flat.
\end{lem}
\begin{proof}
We need to check that for every inner face $F$ of $\GQ$, $C(F)=1$. Inner faces
are of three types, there are rectangles, faces corresponding
to vertices of $\Gs$ and faces corresponding to vertices of the
restricted dual $\Gdres$. 

If $F$ is a rectangle, then by Point 1, the Kasteleyn curvature is,
$
C(F)=(-1)^{2-1}\frac{i^2}{1^2}=~1.
$
If $F$ is a face of length $2k$, corresponding to a vertex of the restricted
dual $\Gdres$, the situation, as well as the labeling of the vertices is
represented in Figure \ref{fig:FigKastFlat1}.

\begin{figure}[ht]
\begin{center}
\psfrag{e1}[c][c]{\scriptsize $e_1$}
\psfrag{e2}[c][c]{\scriptsize $e_2$}
\psfrag{t1}[c][c]{\tiny $\theta_{e_1}$}
\psfrag{t2}[c][c]{\tiny $\theta_{e_2}$}
\psfrag{w1}[c][c]{\scriptsize $\ws_1$}
\psfrag{w2}[c][c]{\scriptsize $\ws_2$}
\psfrag{b1}[c][c]{\scriptsize $\bs_1$}
\psfrag{b2}[c][c]{\scriptsize $\bs_2$}
\psfrag{...}[c][c]{\scriptsize $\cdots$}
\includegraphics[height=4cm]{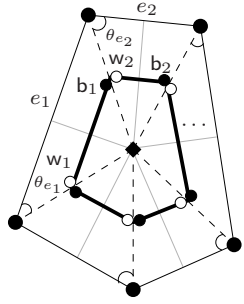}
\caption{Notations for a face of length $2k$, with
$k=5$.}\label{fig:FigKastFlat1}
\end{center}
\end{figure}

By Points 1 and 2(a), the Kasteleyn curvature is:
\begin{equation*}
C(F)=(-1)^{k-1}\frac{1}{e^{i\sum_{j=1}^k\bigl(\frac{3\pi}{2}-\theta_{e_j}\bigr
)}}=-\frac{1}{i^k e^{-i\sum_{j=1}^k\theta_{e_j}}}.
\end{equation*}
Observing that $-\sum_{j=1}^k \theta_{e_j}=-\frac{\pi}{2}k+\pi$, yields
$C(F)=1$.

If $F$ is a face of length $2k$, corresponding to a vertex of the primal graph
$\Gs$, the situation, as well as the labeling of the vertices is represented in
Figure \ref{fig:FigKastFlat2}. 

\begin{figure}[ht]
\begin{center}
\psfrag{eb}[c][c]{\scriptsize $e^{\partial}$}
\psfrag{o}[c][c]{\scriptsize $\rs$}
\psfrag{e1}[c][c]{\scriptsize $e_1$}
\psfrag{e2}[c][c]{\scriptsize $e_2$}
\psfrag{t1}[c][c]{\tiny $\theta_{e_1}$}
\psfrag{t2}[c][c]{\tiny $\theta_{e_2}$}
\psfrag{tb}[c][c]{\tiny $\theta^\partial$}
\psfrag{w1}[c][c]{\scriptsize $\ws_1$}
\psfrag{w2}[c][c]{\scriptsize $\ws_2$}
\psfrag{b1}[c][c]{\scriptsize $\bs_1$}
\psfrag{b2}[c][c]{\scriptsize $\bs_2$}
\psfrag{...}[c][c]{\scriptsize $\cdots$}
\includegraphics[height=4.3cm]{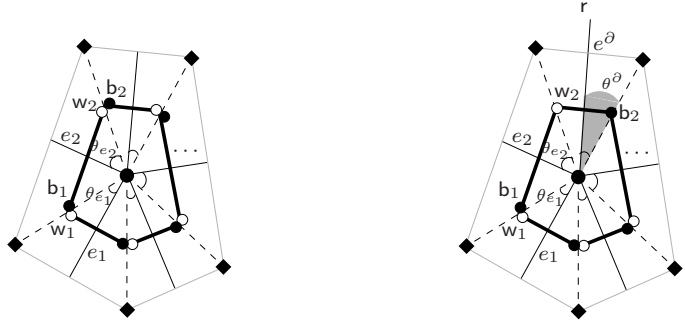}
\caption{Left: notations for a face not incident to the
boundary, with $k=5$. Right: notations for a face
incident to the boundary, with $k=4$.}\label{fig:FigKastFlat2}
\end{center}
\end{figure}

We first consider the case where the face is not
incident to the boundary. By Points 1 and 2(a), the Kasteleyn curvature is:
\begin{equation*}
C(F)=(-1)^{k-1}\frac{e^{i\sum_{j=1}^k\bigl(\frac{3\pi}{2}-\theta_{e_j}\bigr
)}}{i^k}=-\frac{i^k e^{-i\sum_{j=1}^k\theta_{e_j}}}{i^k}.
\end{equation*}
Observing that $\sum_{j=1}^k \theta_{e_j}=\pi$, yields $C(F)=1$. 

We now consider the case where the face is incident to the boundary. By Points
1,
2(a) and 2(b), the Kasteleyn curvature is:
\begin{equation*}
C(F)=(-1)^{k-1}\frac{e^{i\bigl[\sum_{j=1}^k\bigl(\frac{3\pi}{2}-\theta_{e_j
} \bigr
)-\theta^\partial\bigr]}}{i^k}=-\frac{i^k
e^{-i[\sum_{j=1}^k\theta_{e_j}+\theta^\partial]}}{i^k}.
\end{equation*}
Observing that $\sum_{j=1}^k \theta_{e_j}+\theta^\partial=\pi$, yields
$C(F)=1$. 
\end{proof}

\subsection{Conclusion}\label{sec:conclusion}

Let $\Gs$ be an isoradial graph having a
regular isoradial embedding. Consider an Ising model on $\Gs$, with critical coupling constants $\Js$ of Equation~\eqref{equ:IsingWeights}.
Consider the corresponding dimer model on the bipartite graph $\GQ$, with the critical dimer weights $\nu$ of 
Equation~\eqref{equ:KastWeights}. Let $\Ks$ be the Kasteleyn matrix of the graph $\GQ$, with the choice of phases of Section \ref{sec:KastPhase}.
That is, $\Ks$ has lines indexed by white vertices of $\GQ$, columns by black ones. The coefficient 
$\Ks_{\ws,\bs}$ is equal to 0 if $\ws\bs$ is 
not an edge of $\GQ$. If $\ws\bs$ is an edge $\es=\ws\bs$ of $\GQ$, it is equal to:
\begin{equation}\label{equ:KastMatrix1}
\Ks_{\ws,\bs}=
\begin{cases}
\sin\theta_e&\text{ if $\es$ crosses the dual edge $e^*$ of an edge $e$ of $\Gs$,}\\
i\cos\theta_e&\text{ if $\es$ crosses a primal edge $e$ of $\Gs$,}\\
e^{i(\frac{3\pi}{2}-\theta_{e(\ws)})}&\text{ if $\es$ is an external edge, not on the boundary of $\GQ$}\\
e^{i[\frac{3\pi}{2}-(\theta_{e(\ws)}+\theta^\partial)]}&\text{ if $\es$ is an external edge 
on the boundary of $\GQ$}.
\end{cases}
\end{equation}

Then, as a consequence of Theorem \ref{thm:Kuperberg} and Lemma \ref{lem:phasing}, we have:
\begin{prop}\label{prop:dimerK}
The dimer partition function of the graph $\GQ$ with critical weight function $\nu$, is equal to:
\begin{equation*}
\Zdimer(\GQ,\nu)=|\det \Ks|.
\end{equation*}
\end{prop}

\begin{rem}$\,$
\begin{itemize}
\item Note that the determinant of $\Ks$ is real. In general, assigning phases to edges might induce a complex phase factoring out
in the determinant. The complex phase is that of any of the dimer configurations. Since the contribution of the
dimer configuration of $\GQ$ consisting of all edges crossing dual edges of $\Gs$ is real, we know that the determinant is real.
\item We have used the near-isoradial embedding of the graph $\GQ$ to define a flat phasing of the edges. Now that this is done,
it will be more convenient to consider a planar embedding of $\GQ$ such that inner faces correspond, either to crossings of primal and dual
vertices of $\Gs$ and $\Gs^*$, or to vertices of the graph $\Gs$, or to vertices of the restricted dual $\Gs^*$. 
The outer face corresponds to a vertex of the dual graph $\Gs^*$, or equivalently to the dual vertices of the extended dual graph
$\Gdext$ arising from the splitting of this vertex.
\end{itemize}
\end{rem}

\section{The Kasteleyn matrix as a Laplacian matrix}\label{sec:5}

In Section \ref{sec:ColumnSum}, we compute the sum of the columns of the Kasteleyn matrix $\Ks$ of the graph~$\GQ$.
Using this information and a reinterpretation of the vertices of $\GQ$ this allows us, in Section~\ref{sec:KastLapl}, to 
interpret the Kasteleyn matrix $\Ks$ as the Laplacian matrix
of a directed graph $\GOo$ constructed from $\GQ$. Using the matrix-tree theorem for directed graphs, we deduce that the
determinant of the matrix $\Ks$ counts oriented spanning trees of the directed graph $\GOo$. In Section \ref{sec:onemore}, 
we modify the directed graph $\GOo$ along the boundary. This is key to being able to handle the boundary of the graph.

\subsection{Column sum of the Kasteleyn matrix}\label{sec:ColumnSum}

Consider the bipartite graph $\GQ$. Recall that all vertices of $\GQ$ have degree 3;
given a white vertex $\ws$, let us denote its three neighbors $\bs_1,\bs_2,\bs_3$ in counterclockwise
order, where $\ws\bs_3$ denotes the external edge of $\GQ$, see Figure \ref{fig:SumOfWeights}. 

\begin{figure}[ht]
\begin{center}
\psfrag{o}[c][c]{\scriptsize $\rs$}
\psfrag{e1}[c][c]{\scriptsize $e_1$}
\psfrag{eb}[l][l]{\scriptsize $e^\partial$}
\psfrag{e}[c][c]{\scriptsize $e$}
\psfrag{te}[c][c]{\tiny $\theta_{e}$}
\psfrag{tb}[c][c]{\tiny $\theta^\partial$}
\psfrag{w}[c][c]{\scriptsize $\ws$}
\psfrag{b1}[c][c]{\scriptsize $\bs_1$}
\psfrag{b2}[c][c]{\scriptsize $\bs_2$}
\psfrag{b3}[c][c]{\scriptsize $\bs_3$}
\includegraphics[height=2cm]{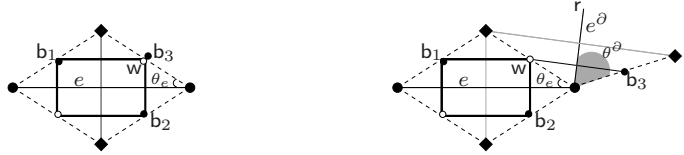}
\caption{Labeling of the neighbors $\bs_1,\bs_2,\bs_3$ of a white vertex $\ws$,
when the vertex $\ws$ is not on the boundary of $\GQ$ (left), and when it is
(right).}\label{fig:SumOfWeights}
\end{center}
\end{figure}

Let $\Ks$ be the Kasteleyn matrix of the graph $\GQ$ with critical dimer weights $\nu$, see Section \ref{sec:conclusion} and
Equation \eqref{equ:KastMatrix1} for definition. The following lemma computes the sum of the columns of the Kasteleyn matrix.

\begin{lem}\label{lem:sumofcolumns}
The sum of the columns of the Kasteleyn matrix $\Ks$ corresponding to a white vertex $\ws$ of $\GQ$ is equal to:
\begin{equation*}
\sum_{i=1}^3\Ks_{\ws,\bs_i}=
\begin{cases}
-ie^{-i\theta_e}(e^{-i\theta^\partial}-1) &\text{ if $\ws$ is a boundary vertex of $\GQ$}\\
0& \text{ otherwise}.
\end{cases}
\end{equation*}
\end{lem}
\begin{proof}
Suppose first that $\ws$ is not a boundary vertex of $\GQ$. Then, by definition of the Kasteleyn matrix $\Ks$, we have:
\begin{align}\label{equ:mw1}
\sum_{j=1}^3\Ks_{\ws,\bs_j}&=\sin\theta_{e}+i\cos\theta_e+e^{
i\bigl(\frac { 3\pi } {2} -\theta_e\bigr)}\nonumber\\\nonumber
&=\frac{i}{2}[-e^{i\theta_e}+e^{-i\theta_e}+
e^{i\theta_e}+e^{-i\theta_e}-2e^{-i\theta_e}
]\nonumber\\
&=0.
\end{align}

Suppose now that $\ws$ is a vertex on the boundary of $\GQ$. In this case, the
phase of the edge $\ws\bs_3$ changes, and we have:

\begin{align}\label{equ:mw2}
\sum_{j=1}^3\Ks_{\ws,\bs_j}&=\sin\theta_e+i\cos\theta_e+e^{
i\bigl(\frac{3\pi
} {2} -\theta_e-\theta^\partial\bigr)}\nonumber\\
&=\frac{i}{2}[-e^{i\theta_e}+e^{-i\theta_e}+
e^{i\theta_e}+e^{-i\theta_e}-2e^{-i(\theta_e-\theta^\partial)}
]\nonumber\\
&=ie^{-i\theta_e}(1-e^{-i\theta^\partial}).
\end{align}
\end{proof}

As a consequence of Lemma \ref{lem:sumofcolumns}, the coefficient $\Ks_{\ws,\bs_3}$ of the Kasteleyn matrix can be rewritten as:
\begin{equation}\label{equ:Kdiag}
\Ks_{\ws,\bs_3}=
\begin{cases}
-\Ks_{\ws,\bs_1}-\Ks_{\ws,\bs_2}-ie^{-i\theta_e}(e^{-i\theta^\partial}-1) &\text{ if $\ws$ is a boundary vertex of $\GQ$}\\
-\Ks_{\ws,\bs_1}-\Ks_{\ws,\bs_2}& \text{ otherwise}.
\end{cases}
\end{equation}

\subsection{Writing the Kasteleyn matrix as a Laplacian matrix}\label{sec:KastLapl}

Using Lemma \ref{lem:sumofcolumns}, we interpret the Kasteleyn matrix of the graph $\GQ$
as the Laplacian matrix of a related, directed graph.
Let us order the lines and columns of the matrix $\Ks$ in such a way that if
the $i$-th line corresponds to the white vertex $\ws$, then
the $i$-th column represents the black vertex $\bs_3$. The
coefficient $\Ks_{\ws,\bs_3}$ becomes a diagonal element, and we write it as in
Equation \eqref{equ:Kdiag}. Note that this
transformation only changes the determinant by an overall $\pm$ sign. Now, instead of labeling the line by $\ws$ and 
the column by $\bs_3$, we label it by a common vertex, denoted $\xs$. 

Geometrically, consider the directed graph $\GOo$ obtained from the graph $\GQ$ as follows.
\begin{itemize}
\item Replace edges of quadrangles of $\GQ$ by oriented ones, from the white
vertex to the black vertex. An oriented edge $(\xs,\ys)$ of a quadrangle is
assigned weight ${\rhoo}_{(\xs,\ys)}=\sin\theta_e$, if it crosses a dual edge $e^*$ of an edge $e$
of $\Gs$; and weight ${\rhoo}_{(\xs,\ys)}=i\cos\theta_e$, if
it crosses a primal edge $e$ of $\Gs$. Note that each quadrangle has two corners with two outgoing edges, and two corners
with two incoming edges.
\item Each external edge $\ws\bs_3$ is merged into a single vertex $\xs$.
\item Each vertex $\xs$ arising
from the merging of an edge $\ws\bs_3$ on the boundary of the outer face of $\GQ$
is joined to the vertex $\rs$, where recall that the vertex $\rs$ is the vertex of the extended graph $\Gsext$ 
corresponding to the outer face of the extended dual $\Gdext$. The oriented edge $(\xs,\rs)$ is assigned weight
${\rhoo}_{(\xs,\rs)}=ie^{-i\theta_e}(1-e^{-i\theta^\partial})$.
\end{itemize}
The directed graph $\GOo$ corresponding to the graph $\GQ$ of Figure \ref{fig:Iso1} is pictured in Figure \ref{fig:Iso4}.
Edges that are dotted are reduced to a
point. A \emph{boundary vertex} $\xs$ of the directed graph $\GOo$ is a vertex arising from
the merging of an edge $\ws\bs_3$ on the boundary of the outer face of $\GQ$. A boundary vertex $\xs$
has out-degree 3 and in-degree 2; an inner vertex $\xs$ has out-degree 2 and in-degree 2.

\begin{figure}[ht]
\begin{center}
\psfrag{r}[l][l]{\scriptsize $\rs$}
\psfrag{x}[l][l]{\scriptsize $\xs$}
\psfrag{sin}[l][l]{\scriptsize $\sin\theta_e$}
\psfrag{cos}[l][l]{\scriptsize $i\cos\theta_e$}
\psfrag{mbw}[l][l]{\scriptsize $ie^{-i\theta_e}(e^{-i\theta^\partial}-1)$}
\psfrag{te}[l][l]{\scriptsize $\theta_e$}
\psfrag{tb}[l][l]{\scriptsize $\theta^\partial$}
\psfrag{e}[l][l]{\scriptsize $e$}
\psfrag{G0}[l][l]{$\GOo$}
\psfrag{l1}[l][l]{\scriptsize Weights at a boundary vertex $\xs$}
\psfrag{l2}[l][l]{\scriptsize Weights at a non-boundary vertex $\xs$}
\includegraphics[width=\linewidth]{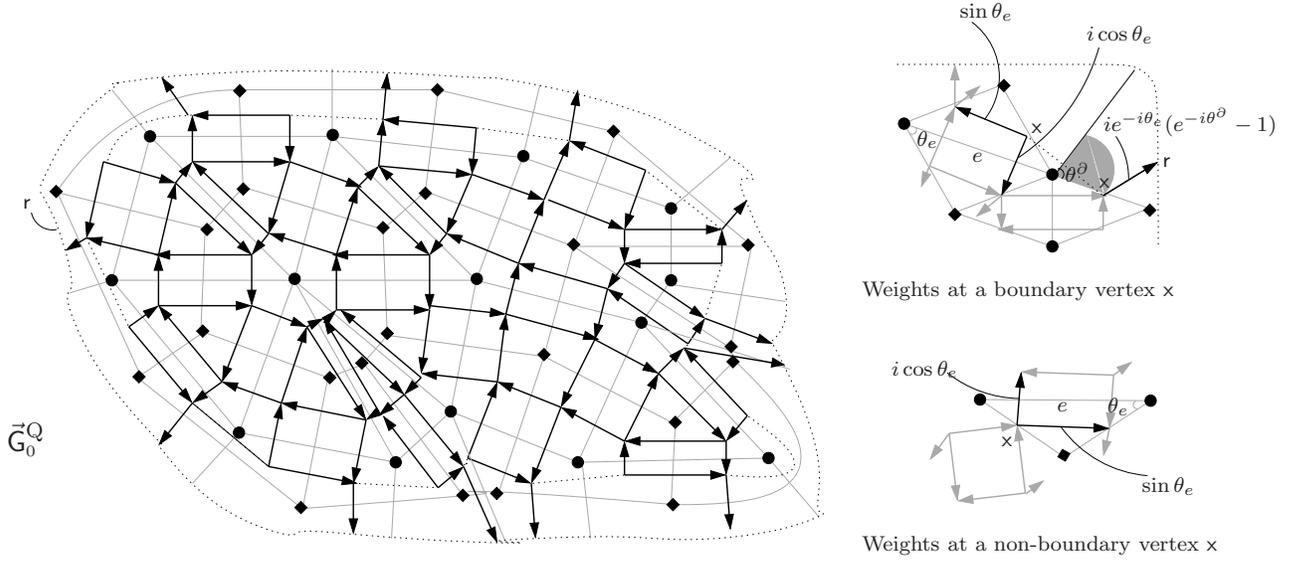}
\caption{Left: the directed graph $\GOo$ obtained from the bipartite graph $\GQ$. Right: weight $\rhoo$ assigned to edges.}
\label{fig:Iso4}
\end{center}
\end{figure}

In accordance with the notation introduced in Section \ref{sec:LaplacianDirected}, we let $\Delta_{\GOo}$ be the
Laplacian matrix of the directed graph $\GOo$ with weight function $\rhoo$
assigned to edges. Observing that the Kasteleyn matrix $\Ks$ is exactly the matrix $\Delta_{\GOo}$ from which the 
line and the column corresponding to the vertex $\rs$ have been removed we deduce, from Theorem \ref{thm:matrixtree}, 
the following proposition. 

\begin{prop}\label{prop:DimerTrees1}
The dimer partition function of the graph $\GQ$ with weight function $\nu$, is equal to 
the absolute value of the $^{\rs}\mathrm{OSTs}$ partition
function of the graph $\GOo$ with weight function $\rhoo$:
\begin{equation*}
\Zdimer(\Gs,\nu)=|\ZOST^{\rs}(\GOo,\rhoo)|.
\end{equation*}
\end{prop}

\subsection{One more graph transformation}\label{sec:onemore}

From the directed graph $\GOo$ with weight function $\rhoo$, we construct a directed graph $\GO$ 
with weight function $\rho$ on the edges, such that their respective $^{\rs}\mathrm{OSTs}$ partition functions are equal.
The directed graph $\GO$ only differs from the graph $\GOo$ along the boundary. The purpose of this transformation
is to have a directed graph with in and out degree two at every vertex, and to have half-quadrangles along the boundary, with
one corner and two outgoing edges. This is key to being able to handle the boundary of the graph.

\begin{itemize}
\item For every boundary vertex $\xs$ of $\GOo$ arising from the merging of a boundary edge $\ws\bs_3$ of $\GQ$, 
undo the merging procedure. That is, we replace
the vertex $\xs$ by the two vertices $\ws$ and $\bs_3$,
add the oriented edge $(\bs_3,\ws)$ and assign it weight,
\begin{equation*}
\rho_{(\bs_3,\ws)}=1.
\end{equation*}
The edge $(\xs,\rs)$ exiting the
vertex $\xs$ is drawn to exit the vertex $\bs_3$, and the edge  $(\bs_3,\rs)$ is assigned weight
\begin{equation*}
\rho_{(\bs_3,\rs)}=e^{-i\theta^\partial}-1.
\end{equation*}
\item All other vertices, edges and edge-weights are left unchanged:
\begin{equation*}
\rho_{(\xs,\ys)}={\rhoo}_{(\xs,\ys)}.
\end{equation*}
\end{itemize}

The transformation is pictured in Figure \ref{fig:FigWeight2}. 

\begin{figure}[ht]
\begin{center}
\psfrag{r}[l][l]{\scriptsize $\rs$}
\psfrag{x}[l][l]{\scriptsize $\xs$}
\psfrag{w}[c][c]{\scriptsize $\ws$}
\psfrag{b3}[c][c]{\scriptsize $\bs_3$}
\psfrag{sin}[l][l]{\scriptsize $\sin\theta_e$}
\psfrag{cos}[l][l]{\scriptsize $i\cos\theta_e$}
\psfrag{mbw}[l][l]{\scriptsize $ie^{-i\theta_e}(e^{-i\theta^\partial}-1)$}
\psfrag{te}[l][l]{\scriptsize $\theta_e$}
\psfrag{tb}[l][l]{\scriptsize $\theta^\partial$}
\psfrag{e}[l][l]{\scriptsize $e$}
\psfrag{mb1}[l][l]{\scriptsize $e^{-i\theta^\partial}-1$}
\psfrag{G0}[l][l]{$\GOo$}
\psfrag{G}[l][l]{$\GO$}
\includegraphics[width=10cm]{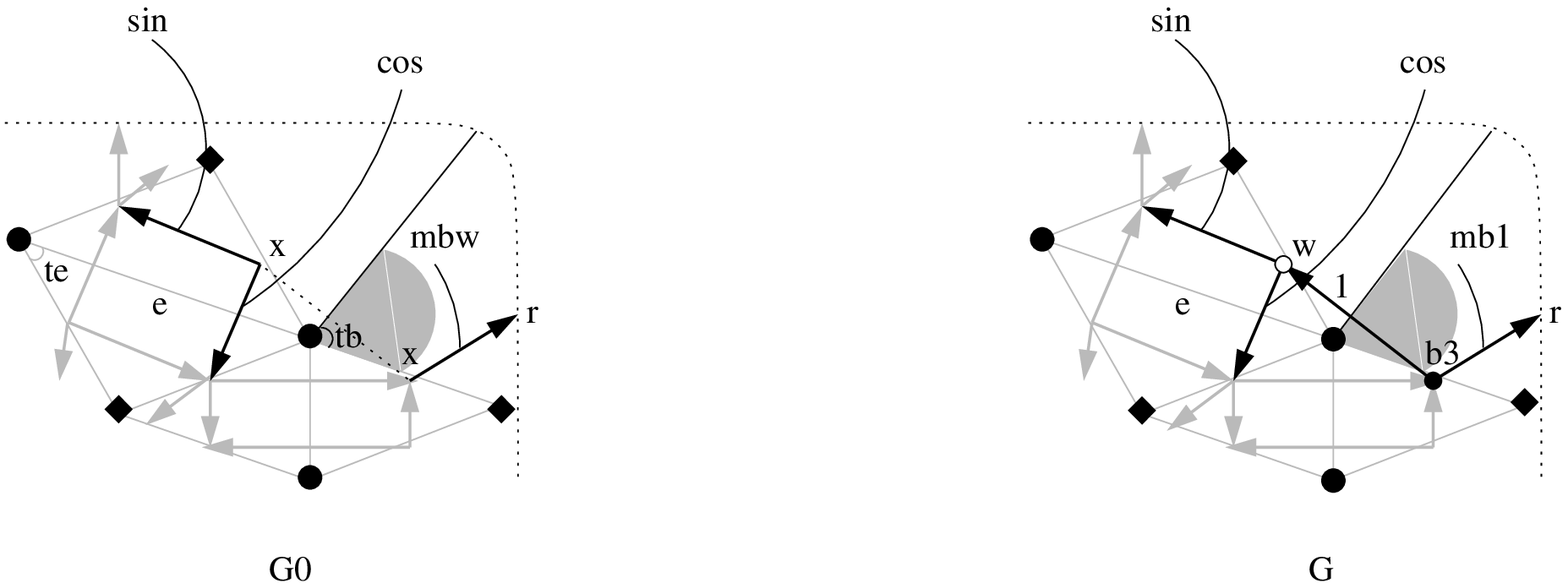}
\caption{From the graph $\GOo$ to the graph $\GO$.}\label{fig:FigWeight2}
\end{center}
\end{figure}

Observe that boundary vertices now have in and out degree equal to two, and that the edges $(\bs_3,\ws)$, $(\bs_3,\rs)$ can 
be seen as forming one corner of a quadrangle. The next proposition proves that this transformation preserves 
$^{\rs}\mathrm{OSTs}$ partition functions.
\begin{prop}\label{prop:treestrees}
The $^{\rs}\mathrm{OSTs}$ partition function of the graph $\GO$ with weight function $\rho$ is equal to the 
$^{\rs}\mathrm{OSTs}$ partition function of the graph $\GOo$ with weight function $\rhoo$:
\begin{equation*}
\ZOST^{\rs}(\GO,\rho)=\ZOST^{\rs}(\GOo,\rhoo).
\end{equation*}
\end{prop}

\begin{proof}

The proof consists in providing a weight-preserving mapping between $^{\rs}\mathrm{OSTs}$ of $\GOo$ and 
$^{\rs}\mathrm{OSTs}$ of $\GO$.
Recall that an OST is characterized by the fact that it has one outgoing edge at every vertex except the root $\rs$, and that it
contains no cycle. We start by exhibiting a weight-preserving mapping between oriented edge configurations of $\GOo$ and $\GO$ 
having one outgoing edge at every vertex except the root, and then show that this mapping preserves the property of not having cycle.

In the graph $\GOo$, a boundary vertex $\xs$ has out-degree three. As a consequence,
an oriented edge configuration of $\GOo$ having one outgoing edge at every vertex except the root, is  
locally one of the three configurations A1-A2-A3. In a similar way,
considering the two corresponding vertices $\bs_3$ and $\ws$ of $\GO$, each having out-degree two, an 
oriented edge configuration of $\GO$ having one outgoing edge at every vertex except the root, is  
locally one of the four configurations D1-D2-D3a-D3b. Consider the mapping between
such configurations, which does not change the configuration at non-boundary
vertices, and which acts as in Figure \ref{fig:FigWeight} (line A - line D)
for boundary vertices, \emph{i.e} A1 is mapped to D1, A2 is mapped to D2, A3 is mapped to $\{\mathrm{D3a,D3b}\}$.

\begin{figure}[ht]
\begin{center}
\psfrag{A1}[c][c]{\scriptsize A1}
\psfrag{A2}[c][c]{\scriptsize A2}
\psfrag{A3}[c][c]{\scriptsize A3}
\psfrag{B3a}[c][c]{\scriptsize B3}
\psfrag{C3a}[c][c]{\scriptsize C3a}
\psfrag{C3b}[c][c]{\scriptsize C3b}
\psfrag{D1}[c][c]{\scriptsize D1}
\psfrag{D2}[c][c]{\scriptsize D2}
\psfrag{D3a}[c][c]{\scriptsize D3a}
\psfrag{D3b}[c][c]{\scriptsize D3b}
\psfrag{x}[c][c]{\scriptsize $\xs$}
\psfrag{w}[c][c]{\scriptsize $\ws$}
\psfrag{b3}[c][c]{\scriptsize $\bs_3$}
\includegraphics[width=\linewidth]{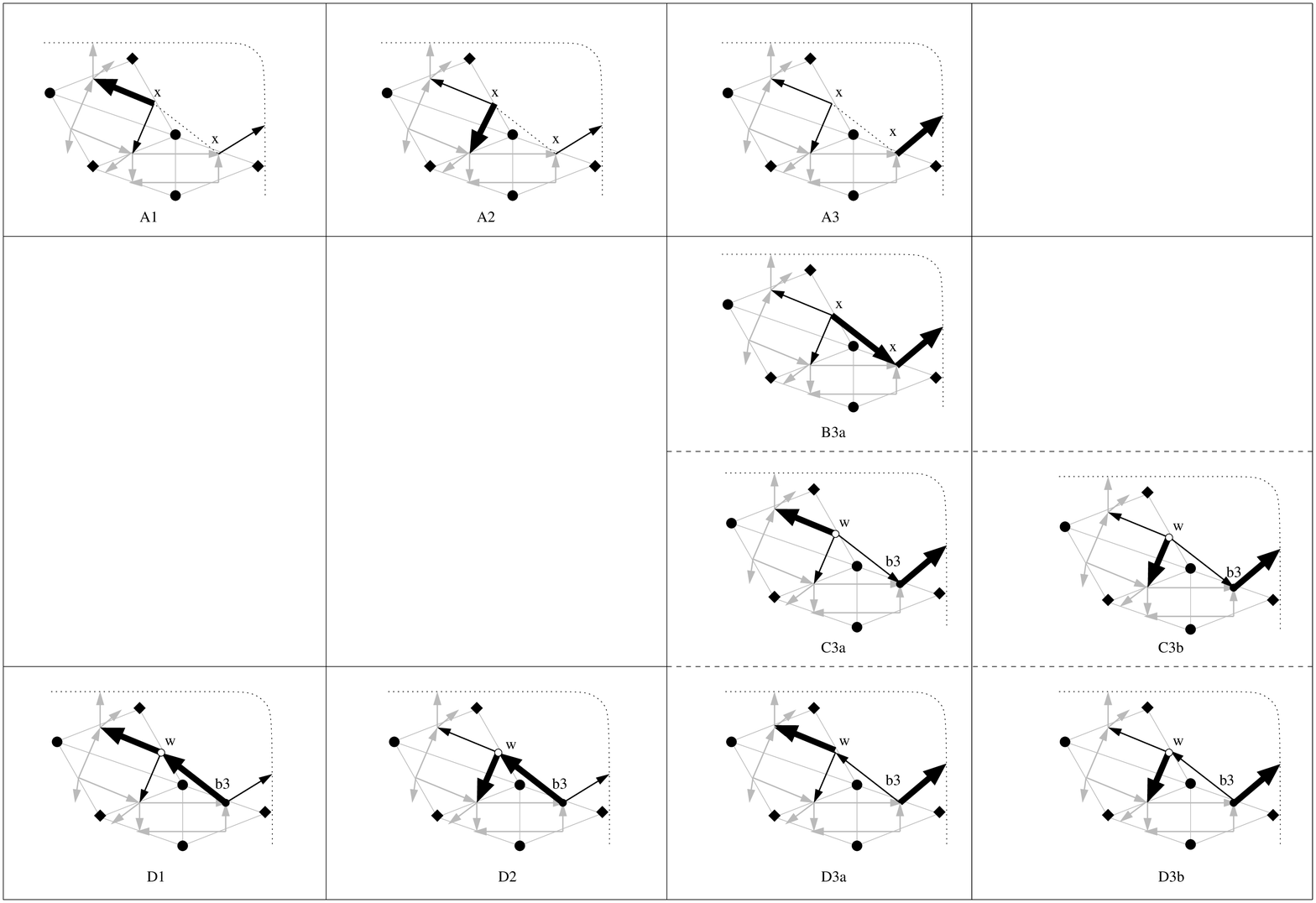}
\caption{From $^{\rs}\mathrm{OSTs}$ of $\GOo$ to $^{\rs}\mathrm{OSTs}$ of $\GO$.}\label{fig:FigWeight}
\end{center}
\end{figure}

Let us prove that if the weight of such a configuration is the product of its edge-weights, then the mapping is
weight-preserving.

The mapping leaves edges exiting non-boundary vertices unchanged, and by definition $\rhoo$ and $\rho$
are equal on such edges, implying that the contributions are equal.
It remains to check that the contributions of edges exiting
boundary vertices are the same. Let us fix a boundary vertex $\xs$ of $\GOo$ and the corresponding boundary vertices 
$\ws,\bs_3$ of $\GO$. By definition, the weight function $\rho$ assigns 
weight 1 to the edge $(\bs_3,\ws)$, so that the weight is preserved for
Cases 1 and 2. Let us check Case 3. The
contribution of the edge exiting $\xs$ is:
$$
{\rhoo}_{(\xs,\rs)}=ie^{-i\theta_e}(e^{-i\theta^{\partial}}-1).
$$
By definition of the weight function $\rho$, the contribution at $\ws$ and $\bs_3$ of the corresponding
configurations of $\GO$ is:
$$
\rho_{(\bs_3,\rs)}[\rho_{(\ws,\bs_1)}+\rho_{(\ws,\bs_2)}]=
(e^{-i\theta^\partial}-1)
(\sin\theta_e+i\cos\theta_e)=(e^{-i\theta^\partial}-1)ie^{
-i\theta_e}.
$$
As a consequence, the mapping is weight-preserving.

Let us now prove that the mapping preserves OSTs, \emph{i.e.} let us show that
an oriented edge configuration of $\GOo$ is an $^{\rs}\mathrm{OSTs}$ of $\GOo$ if and only the
oriented edge configuration(s) obtained from the mapping is (are) OST(s) of
$^{\rs}\mathrm{OSTs}$ of $\GO$. Since the mapping preserves the property of having one outgoing edge at every vertex except the root,
this amounts to showing that the mapping preserves the property of not containing cycles. 

Cases 1 and 2 are clear because the vertex $\xs$ has been
split into two vertices and only the edge joining those two vertices is added. 
To prove Case 3, we use the intermediate lines B3 and C3a-C3b of 
Figure \ref{fig:FigWeight}. The statement between 
A3 and B3 is clear for the same reason as Cases 1 and 2. The statement holds between C3a and D3a (C3b and D3b) 
because the oriented edge configurations and the unoriented versions of the underlying graphs are the same,
what changes is the orientation of one unused edge of the graph. So we are left with proving that the statement holds 
between B3 and C3a-C3b. We only prove it for B3 and C3a, since the argument is the same for B3 and C3b.
The configurations B3 and C3a are defined on the same directed graph and differ at a single edge. Let us denote by $\es_{\text{\tiny{B}}}$, respectively 
$\es_{\text{\tiny{C}}}$, the edge contained in the configuration B3 and not C3a (resp. in C3a and not B3); and let 
$\ws$ be their common vertex.

The configuration C3a is obtained from the configuration B3 by adding the edge $\es_{\text{\tiny{C}}}$ and then removing the edge 
$\es_{\text{\tiny{B}}}$. Consider the oriented configuration $\Fs$ obtained by adding $\es_{\text{\tiny{C}}}$ to B3, then since
every vertex except the root of B3 and C3a has one outgoing edge, the oriented configuration $\Fs$ has outdegree two at the common
vertex $\ws$. Now, suppose that B3 contains no cycle, then $\Fs$ has a single cycle containing $\es_{\text{\tiny{C}}}$. 
When removing the edge $\es_{\text{\tiny{B}}}$, this cycle is broken if 
and only if the edge $\es_{\text{\tiny{B}}}$ is an edge of the cycle. We thus need to show that the edge $\es_{\text{\tiny{B}}}$ 
belongs to the cycle created when adding the edge $\es_{\text{\tiny{C}}}$ to B3. 
The common vertex $\ws$ belongs to the cycle, if 
the cycle does not contain the edge $\es_{\text{\tiny{B}}}$, then the vertex $\ws$ must be incident to three edges. We know that
$\ws$ has outdegree two, meaning that it must have indegree 1. This is not possible, because $\ws$ has only outgoing edges. We 
have thus proved that property of not having cycles is preserved. This argument being symmetric, we have proved an 
`if and only if'.
\end{proof}

From Proposition \ref{prop:DimerTrees1} and Proposition \ref{prop:treestrees}, we deduce the following:
\begin{cor}\label{cor:DimerTrees1}
The dimer partition function of the graph $\GQ$ with weight function $\nu$, is equal to the absolute value of the
$^{\rs}\mathrm{OSTs}$ partition function of the graph $\GO$ with weight function $\rho$:
\begin{equation*}
\Zdimer(\GQ,\nu)=|\ZOST^{\rs}(\GO,\rho)|.
\end{equation*}
\end{cor}

\section{Oriented spanning trees of $\GO$ and of the extended double}\label{sec:6}

In this section, we describe a weight-preserving mapping from OSTs of the graph $\GO$ to a family of spanning
trees of the extended double graph $\GDext$ of $\Gs$, which we then characterize.

\subsection{Dual spanning trees of the extended double graph}\label{sec:dual}

Consider the dual graph $\GOd$ of the undirected version of $\GO$, see for example Figure~\ref{fig:FigIso5} (left).
Eventhough $\GOd$ is not bipartite, its vertices are naturally split into black and white vertices as follows.
\begin{itemize}
\item[-] $\bullet$-black: vertices of the primal graph $\Gs$,
\item[-] $\smalllozenge$-black: vertices of the extended dual graph $\Gdext$.
\item[-] White: dual vertices of the quadrangles of $\GO$.
\end{itemize}

Recall from Section \ref{sec:LaplacianDirected} that if $\Ts$ is a spanning tree of a graph, the dual configuration
$\Ts^*$ consisting of the complement of the edges dual to $\Ts$, is a spanning tree of the dual graph. As a consequence,
there is a one-to-one correspondence between $^{\rs}\mathrm{OSTs}$ of $\GO$ and dual spanning trees of $\GOd$.

The graph $\GOd$ is nearly the graph we are aiming for; we need to make one more transformation so that it becomes the \emph{extended 
double graph} of $\Gs$, defined as follows. Consider the extended dual graph $\Gdext$ of $\Gs$, and the dual graph 
$\Gsext$ of $\Gdext$. The \emph{extended double graph} $\GDext$ of $\Gs$, consists of the graph $\Gdext$ and the graph $\Gsext$, with a white
vertex added at the crossing of each primal and dual edge; one then removes all edges connected to the vertex $\rs$ 
of $\Gsext$ 
corresponding to the outer face of $\Gdext$. An example of extended double graph is given in Figure~\ref{fig:FigIso5} (right).
The construction of the double graph first appears in the paper of Temperley \cite{Temperley}
proving a bijection between spanning trees of the square grid and perfect matchings its the double graph. This bijection
has been generalized to more general planar graphs and oriented spanning trees by Kenyon, Propp and Wilson in \cite{KPW}.

\begin{figure}[ht]
\begin{center}
\includegraphics[width=\linewidth]{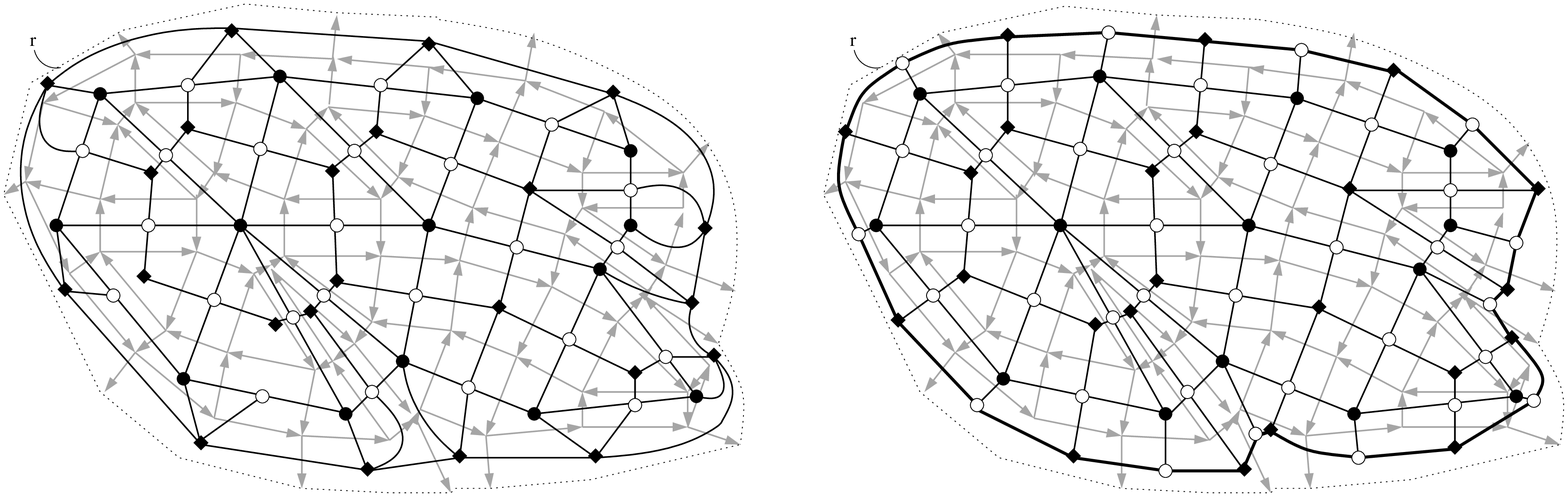}
\caption{Left: the dual graph $\GOd$ (black) of $\GO$ (grey). Right: the extended double graph $\GDext$ of $\Gs$ (black).}\label{fig:FigIso5}
\end{center}
\end{figure}

Observe that the extended double graph $\GDext$ of $\Gs$ is obtained from the dual graph $\GOd$ by
splitting every boundary $\smalllozenge$-black vertex of the extended dual $\Gdext$,
into a $\smalllozenge$-black vertex $\bs$ and a white vertex $\ws$, and by adding the edge $\ws\bs$.
As a consequence, to every spanning tree of $\GOd$ corresponds a spanning tree of the extended double
$\GDext$, obtained by adding the edges arising
from the splitting of the boundary $\smalllozenge$-black vertices. 
This mapping is weight-preserving if the new edges are assigned
weight 1. 

Consider an $^{\rs}\mathrm{OST}$ $\Ts$ of the directed graph $\GO$ and its dual spanning tree in $\GOd$.
Then, the \emph{dual spanning tree in $\GDext$ of $\Ts$} is defined to be the corresponding spanning tree of the extended 
double graph $\GDext$, see Figure \ref{fig:FigIso5a} for an example.

\begin{figure}[ht]
\begin{center}
\includegraphics[width=8cm]{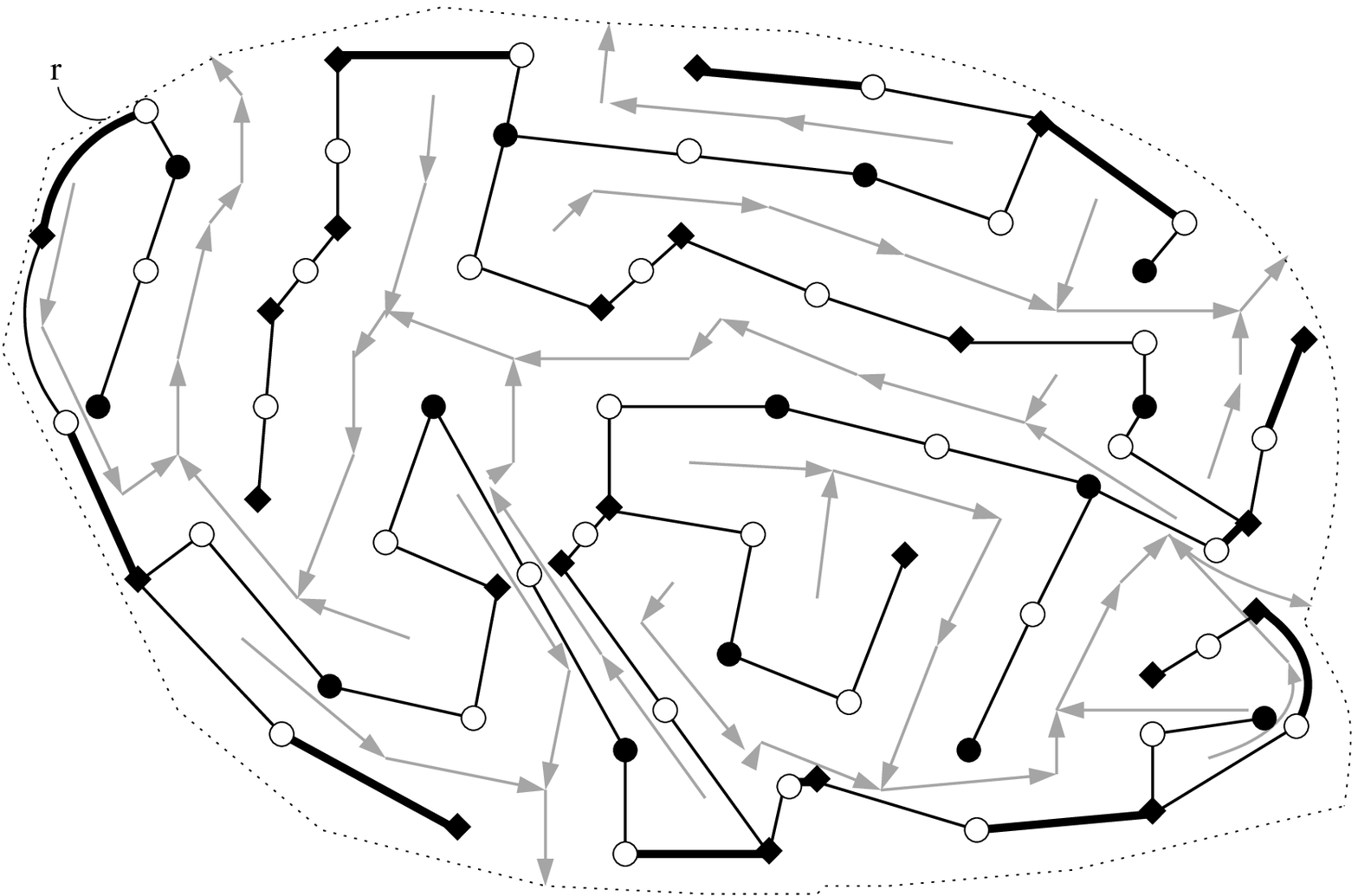}
\caption{An $^{\rs}\mathrm{OST}$ of $\GO$ (grey) and its dual spanning tree in the extended double $\GDext$ graph (black).
By definition of the \emph{dual in $\GDext$}, the thicker edges are always present.}\label{fig:FigIso5a}
\end{center}
\end{figure}

\subsection{Weight preserving mapping}\label{sec:weightpreserving}

Let us now assign weights $\rho^*$ to edges of the double graph $\GDext$, so that the
mapping from $^{\rs}\mathrm{OSTs}$ of $\GO$ with weight function $\rho$, to dual spanning trees in $\GDext$ with 
weight function $\rho^*$, preserves weights.

By construction of dual spanning trees in $\GDext$, we know that edges
arising from the splitting of boundary $\smalllozenge$-black vertices are always present and are
assigned weight 1.

Recall that each vertex $\xs$ of the graph $\GO$, except the root $\rs$, has two outgoing edges; and that an $^{\rs}\mathrm{OST}$ $\Ts$ of $\GO$
takes exactly one of the two edges. Then, the dual of the present (respectively absent) edge in $\Ts$, is absent 
(respectively present) in the dual spanning tree in $\GDext$, see Figure \ref{fig:FigPrimalDual}. This yields a one-to-one correspondence
between edges of $^{\rs}\mathrm{OSTs}$ of $\GO$ and edges of dual spanning trees in $\GDext$. 

\begin{figure}[ht]
\begin{center}
\psfrag{x}[l][l]{\scriptsize $\xs$}
\includegraphics[width=8cm]{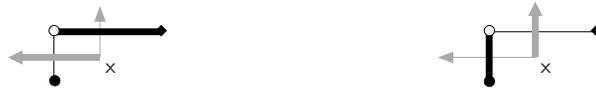}
\caption{Possible primal and dual configurations at a
vertex $\xs$ of $\GO$. Thick lines represent present edges, thin lines absent ones.}\label{fig:FigPrimalDual}
\end{center}
\end{figure}

Returning to the definition of the weight function
$\rho$, see Section \ref{sec:onemore}, this implies that the following weight function 
$\rho^*$ defined on edges of the extended dual $\GDext$, preserves weights, see also Figure~\ref{fig:FigWeight3}:

\begin{equation}\label{def:rhoetoile}
\rho^*_{\ws\bs}=
\begin{cases}
\sin\theta_e&\text{ if $\ws\bs$ is half a primal edge $e$ of $\Gs$}\\
i\cos\theta&\text{ if $\ws\bs$ is half a dual edge $e^*$ of an edge $e$ of $\Gs$}\\
1&\text{ if $\ws\bs$ is half a dual edge on the boundary of the extended dual $\Gdext$}\\
e^{-i\theta^\partial}-1&\text{ if $\ws\bs$ is half an edge of $\Gsext\setminus\Gs$}.
\end{cases}
\end{equation}

\begin{figure}[ht]
\begin{center}
\psfrag{sin}[l][l]{\scriptsize $\sin(\theta_e)$}
\psfrag{cos}[l][l]{\scriptsize $i\cos(\theta_e)$}
\psfrag{te}[l][l]{\scriptsize $\theta_e$}
\psfrag{tb}[l][l]{\scriptsize $\theta^\partial$}
\psfrag{e}[l][l]{\scriptsize $e$}
\psfrag{l1}[l][l]{\scriptsize Weights around a boundary
white vertex}
\psfrag{l2}[l][l]{\scriptsize Weights around a non-boundary
white vertex}
\psfrag{l3}[l][l]{\scriptsize Primal and corresponding dual edge in the dual
configuration}
\psfrag{mb1}[l][l]{\scriptsize $e^{-i\theta^\partial}-1$}
\psfrag{G0}[l][l]{$\vec{\Gs}_0$}
\psfrag{G}[l][l]{$\vec{\Gs}$}
\includegraphics[width=10cm]{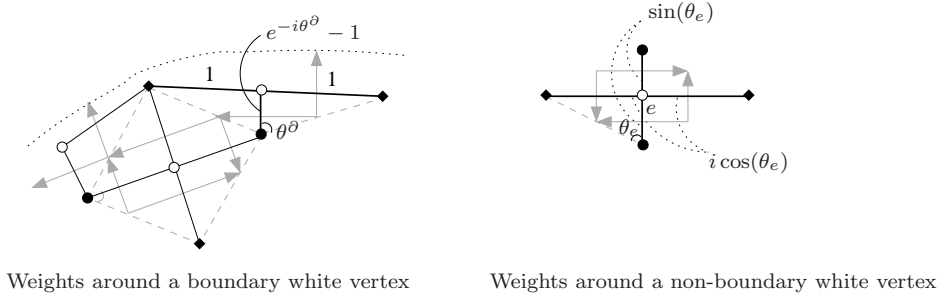}
\caption{Definition of weights $\rho^*$ assigned to edges
of the extended double graph $\GDext$.}\label{fig:FigWeight3}
\end{center}
\end{figure}

\subsection{Characterization of dual spanning trees of the double graph}\label{sec:CharactDual}

The next lemma characterizes spanning trees of $\GDext$ arising as duals of
$^{\rs}\mathrm{OSTs}$ of $\GO$.

\begin{lem}\label{lem:ConfigLoc}
Consider an $^{\rs}\mathrm{OST}$ of $\GO$. Then, the restriction around a white vertex $\ws$ of
$\GDext$ of the dual spanning tree in $\GDext$, is one of the following four
configurations (black lines):
\begin{enumerate}
\item if the vertex $\ws$ is not on the boundary of $\GDext$:

\begin{figure}[ht]
\begin{center}
\includegraphics[width=13cm]{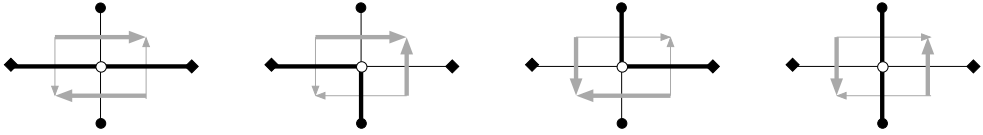}
\caption{The four possible configurations around a non-boundary white vertex of
$\GDext$.}\label{fig:FigDualLocal}
\end{center}
\end{figure}
\item if the vertex $\ws$ is on the boundary of $\GDext$:

\begin{figure}[ht]
\begin{center}
\includegraphics[width=6cm]{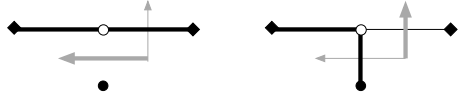}
\caption{The two possible configurations around a boundary white
vertex of $\GDext$.}\label{fig:FigDualLocal1}
\end{center}
\end{figure}
\end{enumerate}
Conversely, consider a spanning tree of $\GDext$ satisfying 1. and 2., 
then the oriented version of its dual spanning tree in $\GO$, oriented
towards the root vertex $\rs$, is an $^{\rs}\mathrm{OST}$ of $\GO$.
\end{lem}
\begin{rem}
Observe that dual spanning trees of $\GDext$, arising as duals of $^{\rs}\mathrm{OSTs}$ of $\GO$, have two edges incident to every
white vertex $\ws$ of $\GDext$. Note also that the two cases of Point 2. are two of the four cases of Point 1. This is because
along the boundary, there are half-quadrangles with one corner and two outgoing edges.
\end{rem}

\begin{proof}
Consider an $^{\rs}\mathrm{OST}$ $\Ts$ of $\GO$, and its dual spanning tree in $\GDext$. Fix a white
vertex $\ws$ of $\GDext$, and assume it is not on the boundary.
There are four edges incident to the vertex $\ws$, and the dual edges form
a quadrangle of $\GO$. Of the four dual edges, two exit one corner of the quadrangle and the other two exit the opposite corner.
Since $\Ts$ is an OST of $\GO$, $\Ts$ contains exactly one edge exiting each of the corners, yielding four 
possible configurations for the restriction of $\Ts$ to the quadrangle. 
Looking at dual configurations, we obtain the four configurations of Point 1.

Now fix a white vertex  $\ws$ on the boundary of $\GDext$. There are three edges
incident to the vertex $\ws$, exactly one of which arises from the splitting of a boundary $\smalllozenge$-black vertex of $\GOd$. 
By definition, this edge always belongs to the dual spanning tree in $\GDext$. The dual of the other two edges belong to $\GO$.
They form a half-quadrangle, and are the two edges exiting the only corner. Since $\Ts$ is an OST, it
contains exactly one edge exiting this corner, yielding two possible configurations for the restriction of $\Ts$ to the half-quadrangle.
Looking at dual configurations, we obtain the two configurations of Point 2. of Figure \ref{fig:FigDualLocal}.

Conversely, suppose we are given a spanning tree of $\GDext$ such that the restriction around every white vertex satisfies Points 1. 
and 2. Taking the dual spanning tree and the corresponding version of
$\GO$ (obtained by merging the boundary edge arising from the splitting of the boundary $\smalllozenge$-black vertex, 
back into a $\smalllozenge$-black vertex) yields a spanning tree of the undirected version of $\GO$. What
remains to check is that the orientation of this spanning tree, induced by the orientation of the edges of $\GO$
yields an $^{\rs}\mathrm{OST}$ of $\GO$. To prove this, it suffices to show that every vertex of
$\GO$, except the root $\rs$, has exactly one outgoing edge. A vertex $\xs$ of
$\GO$ has exactly two outgoing edges, bounding a rectangle (or
half a rectangle) corresponding to a dual vertex $\ws$. Looking at the
possible configurations around the white vertex given by Points 1. and 2., we
know that exactly one of the two edges exiting the vertex $\xs$ is present.
\end{proof}

It is now convenient to root spanning trees of the double graph $\GDext$ at one of the boundary
$\smalllozenge$-black vertex, denoted by $\ss$. Let us denote by
$\T^{\ss}_{\scriptscriptstyle{(1,2)}}(\GDext)$ the set of
spanning trees of the double $\GDext$, rooted at the vertex $\ss$, and satisfying
Points 1. and 2. of Lemma \ref{lem:ConfigLoc}. Then, as a consequence of
Section~\ref{sec:dual}, Section~\ref{sec:weightpreserving} and Lemma~\ref{lem:ConfigLoc}, we have the following proposition.

\begin{prop}\label{prop:TreesDouble}
The $^{\rs}\mathrm{OST}$ partition function of the oriented graph $\GO$ with weight function $\rho$ is equal to the 
partition function of the double graph $\GDext$ with weight function $\rho^*$, restricted to $^{\ss}\mathrm{OST}$ satisfying Points 1. 
and 2. of Lemma \ref{lem:ConfigLoc}.
\begin{equation*}
\ZOST^{\rs}(\GO,\rho)=Z_{\mathrm{OST,(1,2)}}^{\ss}(\GDext,\rho^*).
\end{equation*} 
\end{prop}

\section{From spanning trees of the extended double to spanning trees of $\Gsext$}

We are now done with modifying the graph along the boundary to be able to handle the boundary. This section consists in the
heart of the explicit construction: we define a mapping from perfect matchings of the extended double to spanning trees of the extended double. The proof of Theorem 
\ref{thm:MainResult} is then ended, using a generalized form of Temperley's bijection, due to Kenyon, Propp and Wilson \cite{KPW}.

\subsection{From spanning trees of $\GDext$ to perfect matchings of $\GDext(\ss)$}

Recall that a prefect matching, or dimer configuration of a graph, is a subset of edges such that each vertex of the graph
is incident to exactly one edge of the subset. Let $\ss$ be the root vertex of $\GDext$ chosen in the previous section. Recall
that we have chosen $\ss$ to be a boundary $\smalllozenge$-black vertex. Let us denote by $\GDext(\ss)$ the graph $\GDext$
from which the vertex $\ss$ and all incident edges have been removed. In accordance with the notations introduced, 
$\M(\GDext(\ss))$ denotes the set of perfect matchings of the graph $\GDext(\ss)$.

The next lemma gives a natural way of obtaining a perfect matching of $\GDext(\ss)$ from an $^{\ss}\mathrm{OST}$
of $\GDext$ satisfying Points 1. and 2. of Lemma \ref{lem:ConfigLoc}. An example is provided in Figure \ref{fig:FigIso5b}.

\begin{lem}\label{lem:matching}
Let $\Ts$ be an OST of $\T^{\ss}_{(1,2)}(\GDext)$. Then, the edge
configuration $\Ms_{\Ts}$ of $\GDext(\ss)$, consisting of edges of $\Ts$ exiting from
black vertices, is a perfect matching of $\GDext(\ss)$.
\end{lem}

\begin{figure}[ht]
\begin{center}
\psfrag{s}[l][l]{$\ss$}
\includegraphics[width=\linewidth]{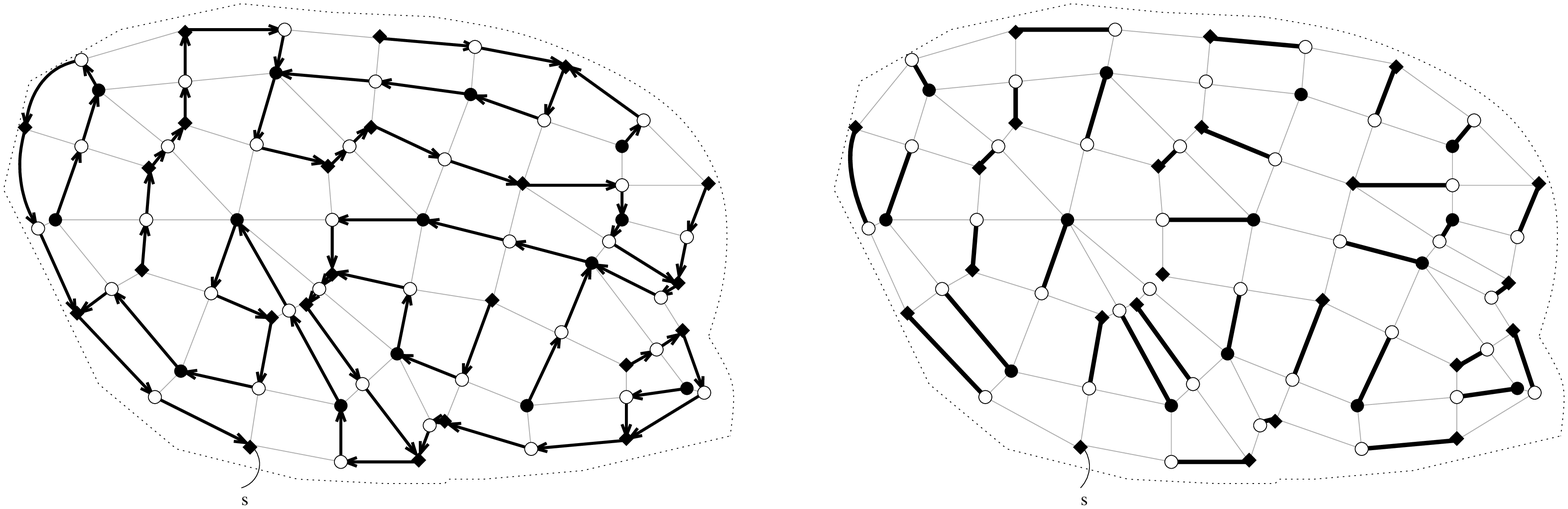}
\caption{Left: oriented spanning tree of $\T^{\ss}_{(1,2)}(\GDext)$. Right: corresponding perfect matching 
$\Ms_{\Ts}$ of $\GDext(\ss)$.}\label{fig:FigIso5b}
\end{center}
\end{figure}

\begin{proof}
Let $\Ts$ be a spanning tree of $\T^{\ss}_{(1,2)}(\GDext)$. We need
to prove that every vertex of $\GDext$ except the root $\ss$ is touched by exactly
one edge of $\Ms_\Ts$. Since $\Ts$ is a spanning tree rooted at $\ss$, every
vertex except the root $\ss$ has exactly one outgoing edge. In particular, this is true
for all black vertices, which are, by definition, part of $\Ms_{\Ts}$. Thus
every black vertex of $\GDext$ except the root is touched by a unique edge of
$\Ms_{\Ts}$. 

Moreover, since $\Ts$ is a spanning
tree $\Ts$ of $\T^{\ss}_{(1,2)}(\GDext)$, it satisfies Points 1. and 2. In
particular, every white vertex has degree exactly 2, with one incoming edge and
one outgoing edge. Since the graph is bipartite, the edge of $\Ts$ entering the
white vertex must exit a black one. By definition, it is thus part
of $\Ms_{\Ts}$, and we conclude that every white vertex is touched by exactly
one edge of $\Ms_{\Ts}$.
\end{proof}

Given a perfect matching $\Ms$ of $\M(\GDext(\ss))$,
a spanning tree $\Ts$ of $\T^{\ss}_{(1,2)}(\GDext)$ is said to be
\emph{compatible with $\Ms$}, if $\Ms_{\Ts}=\Ms$. Let us denote by
$\T^{\ss}_{(1,2),\Ms}(\GDext)$ the set of spanning trees of $\T^{\ss}_{(1,2)}(\GDext)$
compatible with a perfect matching $\Ms$ of $\M(\GDext(\ss))$. Then, we have:

\begin{lem}\label{lem:union}
$$
\T^{\ss}_{(1,2)}(\GDext)=\bigcup_{\Ms\in\M(\GDext(\ss))}\T^{\ss}_{(1,2),\Ms}(\GDext),
$$
and the union is disjoint.
\end{lem}
\begin{proof}
The fact that $\T^{\ss}_{(1,2)}(\GDext)$ is the union of $\T^{\ss}_{(1,2),\Ms}(\GDext)$
is a consequence of Lemma \ref{lem:matching}. Let us prove that the union is
disjoint. Suppose that there are two distinct matchings $\Ms_1$ and $\Ms_2$
of $\M(\GDext(\ss))$, and an OST $\Ts$ which belongs to
$\T^{\ss}_{(1,2),\Ms_1}(\GDext)\cap \T^{\ss}_{(1,2),\Ms_2}(\GDext)$. Then, by definition
of being compatible with $\Ms_1$ and $\Ms_2$, the spanning tree
$\Ts$ must contain all edges of $\Ms_1$ and all edges of $\Ms_2$. But the two
matchings $\Ms_1$ and $\Ms_2$ being distinct, their superimposition contains
a cycle. This yields a contradiction with $\Ts$ being an OST.
\end{proof}

\subsection{Characterization of spanning trees compatible with a matching}

By definition, given a perfect matching $\Ms$ of $\M(\GDext(\ss))$,
the unoriented version of an OST $\Ts$ of
$\T^{\ss}_{(1,2),\Ms}(\GDext)$ satisfies the following two conditions:
it contains all edges of $\Ms$, and it satisfies Points 1.
and 2. of Lemma \ref{lem:ConfigLoc}, at every white vertex of $\GDext$. The next
proposition proves that if an edge configuration satisfies these two conditions,
then it is a spanning tree. This is
a remarkable fact. Indeed, being a spanning tree requires not having cycles,
a non-local condition. Containing edges of a perfect matching $\Ms$
is a non-local condition, but it determines only half of the edges. The point
of the proposition is to show that the configuration of the other half of the edges is
determined locally.

\begin{prop}\label{prop:important}
Let $\Ms$ be a perfect matching of $\M(\GDext(\ss))$. Then, an edge configuration
containing all edges of $\Ms$, and satisfying Points 1. and 2. of Lemma \ref{lem:ConfigLoc}, is the
unoriented version of an OST of $\T^{\ss}_{(1,2),\Ms}(\GDext)$.
\end{prop}
\begin{proof}
Let $\Ts$ be an edge configuration containing all edges of $\Ms$, and satisfying
Points 1. and 2. at every white vertex $\ws$ of $\GDext$. Since the graph $\GDext$
is bipartite, each edge of $\Ms$ joins a black and a white vertex. Orient edges
of $\Ms$ from the black vertex to the white one. 

Since $\Ts$ satisfies Points 1. and 2., every white vertex $\ws$ of $\GDext$ is
incident to exactly two edges of $\Ts$. Exactly one of the two edges is an edge
of $\Ms$, because $\Ms$ is a perfect matching which only leaves the root vertex $\ss$ unmatched, and 
the root vertex $\ss$ is black. By our choice of orientation of $\Ms$,
this edge enters the white vertex. Let us orient the second edge of $\Ts$
incident to $\ws$ away from $\ws$. We have thus defined an oriented version of
the edge configuration $\Ts$ such that every vertex of $\GDext$ except the root
$\ss$ has exactly one outgoing edge, and such that edges exiting black vertices
are exactly those of $\Ms$. As a consequence, if $\Ts$ contains no
cycle, its oriented version is a OST oriented towards $\ss$, satisfying Points
1. and 2., such that $\Ms_{\Ts}=\Ms$, \emph{i.e.} it is an OST of
$\T^{\ss}_{(1,2),\Ms}(\GDext)$.

It thus remains to show that $\Ts$ contains no cycle. Consider the oriented
version of $\Ts$ defined above and suppose it contains a cycle, denoted by
$\Cs$. The graph $\GDext$ being bipartite, vertices of
$\Cs$ alternate between black and white. By our choice of orientation, edges
exiting black vertices belong to $\Ms$, and those exiting white vertices belong
to $\Ts\setminus\Ms$. Moreover, since each vertex of $\GDext$ has
exactly one outgoing edge of $\Ts$, edges of $\Cs$ must be oriented in the same
direction, clockwise or counterclockwise. This implies that edge of the cycle
$\Cs$ alternate between edges of $\Ms$ and edges of $\Ts\setminus\Ms$. Note
that since the root vertex $\ss$ has only incoming edges, it cannot belong to
the cycle, and since it is on the boundary of the graph, it cannot be in the
interior of the cycle.

As a consequence, the edge configuration $\Ms'$ consisting of edges of $\Ms$
away from the cycle $\Cs$, and edges of $\Ts\setminus\Ms$ on the cycle $\Cs$,
is a dimer configuration of $\M(\GDext(\ss))$. The superimposition of $\Ms$ and
$\Ms'$ contains as single cycle, the cycle $\Cs$, and the unmatched vertex $\ss$
is not in the interior of $\Cs$. The interior of $\Cs$ is thus covered by
doubled edges, so that it must contain an \emph{even} number of vertices. 

We now prove that the fact that $\Ts$ satisfies Points 1. and 2. implies that
the cycle $\Cs$ must have an \emph{odd} number of vertices in its interior, thus
yielding a contradiction. We use Euler's formula to prove this. Without loss of
generality, let us suppose that edges of $\Cs$ are oriented clockwise, and
consider the graph ${\GDext}(\Cs)=(\VD(\Cs),\ED(\Cs))$ consisting of the restriction of
$\GDext$ to $\Cs$ and its interior. We now prove that the number of inner vertices
of the graph $\GDext(\Cs)$ is odd. Let us split vertices of $\GDext(\Cs)$ as
follows.
\begin{itemize}
\item[-] Type 1: white vertices of $\Cs$ preceded and followed by a black vertex
of the same type, either $\smalllozenge$-black or $\bullet$-black. Every white
vertex of Type 1 is incident to 3 edges and is on the boundary of 2 faces of
$\GDext(\Cs)$, see Figure \ref{fig:FigDualLocal2}.

\begin{figure}[ht]
\begin{center}
\psfrag{C}[l][l]{\scriptsize $\Cs$}
\psfrag{intC}[r][r]{\scriptsize int$(\Cs)$}
\includegraphics[width=3cm]{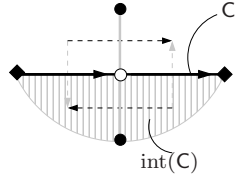}
\caption{White vertex preceded and followed by a $\smalllozenge$-black vertex.}\label{fig:FigDualLocal2}
\end{center}
\end{figure}
\item[-] Type 2: white vertices of $\Cs$ preceded by a $\smalllozenge$-black
vertex and followed by $\bullet$-black vertex. Since the edge configuration
$\Ts$ satisfies Points 1. and 2., only a right turn can occur at such a white
vertex, see Figure \ref{fig:FigDualLocal3}. Every white vertex of Type $2$ is
thus incident to 3 edges and is on the boundary of one face of $\GDext(\Cs)$. 

\begin{figure}[ht]
\begin{center}
\psfrag{C}[l][l]{\scriptsize $\Cs$}
\psfrag{intC}[r][r]{\scriptsize int$(\Cs)$}
\includegraphics[width=3cm]{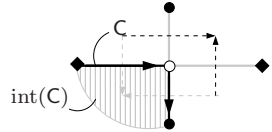}
\caption{White vertices preceded by a $\smalllozenge$-black
vertex and followed by $\bullet$-black vertex}\label{fig:FigDualLocal3}
\end{center}
\end{figure}
\item[-] Type 3: white vertices of $\Cs$ preceded by a $\bullet$-black vertex
and followed by $\smalllozenge$-black vertex. By Points 1. and 2., only a
left turn can occur at such a white vertex, see Figure \ref{fig:FigDualLocal4}.
Every white vertex of Type 3 is thus incident to four edges and is on the boundary
of three faces of $\GDext(\Cs)$.
\begin{figure}[ht]
\begin{center}
\psfrag{C}[l][l]{\scriptsize $\Cs$}
\psfrag{intC}[l][l]{\scriptsize int$(\Cs)$}
\includegraphics[width=3cm]{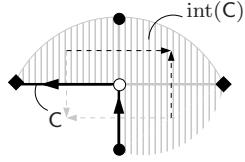}
\caption{White vertices preceded by a $\bullet$-black vertex
and followed by $\smalllozenge$-black vertex.}\label{fig:FigDualLocal4}
\end{center}
\end{figure}
\item[-] Type 4: white vertices in the interior of $\GDext(\Cs)$. Every white
vertex of Type 4 is incident to 4 edges and is on the boundary of 4 faces of
$\GDext(\Cs)$.
\item[-] Type 5: black vertices in the interior of $\GDext(\Cs)$.
\end{itemize}
For every Type $i$, let $n_i$ denote the number of vertices of Type $i$.
The total number $|\VD(\Cs)|$ of vertices of the graph $\GDext(\Cs)$ is equal to
twice the number of boundary white vertices plus the number of inner white and
black vertices, that is:
\begin{equation}\label{equ:vertices}
|\VD(\Cs)|=2(n_1+n_2+n_3)+n_4+n_5. 
\end{equation}
Since the graph $\GDext(\Cs)$ is bipartite, counting the number of edges
$|\ED(\Cs)|$ amounts to counting the number of edges incident to white
vertices. By the description of the different types, we have:
\begin{equation}\label{equ:edges}
|\ED(\Cs)|=3n_1+2n_2+4n_3+4n_4.
\end{equation}
Faces of the graph $\GDext(\Cs)$ are bounded by a cycle of length 4, containing
exactly two white vertices. Thus, summing the number of faces that white
vertices are on the boundary of, yields twice the number of inner faces
$|\FD_{\Cs}|$ of $\GDext(\Cs)$:
\begin{equation}\label{equ:faces}
|\FD_{\Cs}|=\frac{1}{2}(2n_1+n_2+3n_3+4n_4). 
\end{equation}
Euler's formula states that $|\VD(\Cs)|+|\FD_{\Cs}|-|\ED(\Cs)|=1$. Plugging
Equations \eqref{equ:vertices}, \eqref{equ:edges}, \eqref{equ:faces} yields:
\begin{equation*}
\frac{1}{2}(n_2-n_3)-(n_4-n_5)=1. 
\end{equation*}
Moreover, since $\Cs$ is a cycle, $n_2=n_3$, implying that: $n_5-n_4=1$.
Observing that the number of inner vertices of the graph $\GDext(\Cs)$ is
$n_4+n_5$, we deduce that this number is odd, thus concluding the proof.
\end{proof}

\subsection{Weights of spanning trees compatible with a perfect matching.}

Let $\Ms$ be a perfect matching of $\M(\GDext(\ss))$. Using Proposition
\ref{prop:important}, and the definition of the weight function $\rho^*$, see Equation \eqref{def:rhoetoile},
we now compute the $\rho^*$-weighted
sum $\rho^*(\T^{\ss}_{(1,2),\Ms}(\GDext))$ of spanning trees of the extended double graph $\GDext$, satisfying Points 1. and 2. 
of Lemma \ref{lem:ConfigLoc}, compatible with $\Ms$. Since the graph is bipartite, the
contribution is the product of the contribution of edges incident to white
vertices. Recall that they all have degree exactly 2. 
Given a white vertex $\ws$ of $\GDext$, the edge entering
the white vertex is that of the matching $\Ms$ and is common to all spanning
trees of $\T^{\ss}_{(1,2),\Ms}(\GDext)$, let us denote it by $\es_\ws$.

\begin{itemize}
 \item If the white vertex $\ws$ is not on the boundary of $\GDext$, the spanning tree satisfies Point 1. of Lemma \ref{lem:ConfigLoc}
 at $\ws$. Having the edge
$\es_\ws$ leaves only two possible configurations; and by Proposition~\ref{prop:important}, 
the two are possible. Thus the contribution is:

\begin{equation}\label{equ:500}
\rho^*_{\es_{\ws}}(i\cos\theta_e+\sin\theta_e)=ie^{-i\theta_e}\,\rho^*_{\es_{\ws}}.
\end{equation}
\item If the white vertex $\ws$ is on the boundary of $\GDext$, the spanning tree satisfies Point 2. of Lemma \ref{lem:ConfigLoc} at $\ws$.
There are three cases to consider. If the edge $\es_\ws$ is half the dual edge of $\Gdext$ arising from the splitting of a boundary 
$\smalllozenge$-black vertex,
it contributes $1$. By Proposition \ref{prop:important}, the two configurations of Point 2. are possible. Thus the contribution is:
\begin{equation}\label{equ:501}
1\cdot(1+e^{-i\theta^\partial}-1)=e^{-i\theta^\partial}=(-ie^{-i\frac{\theta^\partial}{2}})(ie^{-i\frac{\theta^\partial}{2}}).
\end{equation}
If the edge $\es_\ws$ is the other half dual edge, it contributes $1$. By Point 2, only the first configuration is possible, and 
by Proposition \ref{prop:important}, it is indeed possible. Thus, the contribution is:
\begin{equation}\label{equ:502}
1=(-ie^{-i\frac{\theta^\partial}{2}})(ie^{i\frac{\theta^\partial}{2}}).
\end{equation}
If the edge $\es_\ws$ is half the primal edge of $\Gsext$ incident to $\rs$, 
it contributes $e^{-i\theta^\partial}-1$. Only the second configuration of Point 2.
is possible, and by Proposition \ref{prop:important} it indeed is. Thus, the contribution is:
\begin{equation}\label{equ:503}
e^{-i\theta^\partial}-1=(-ie^{-i\frac{\theta^\partial}{2}})\Bigl(2\sin\frac{\theta^\partial}{2}\Bigr).
\end{equation}
\end{itemize}

Let us assign the following weight function $\tau$ to edges of the extended double graph $\GDext(\ss)$, 
see also Figure~\ref{fig:FigWeight4}.

\begin{equation}\label{equ:weighttau2}
\tau_{\ws\bs}=
\begin{cases}
\sin\theta_e&\text{ if $\ws\bs$ is half a primal edge $e$ of $\Gs$}\\
i\cos\theta_e&\text{ if $\ws\bs$ is half a dual edge $e^*$ of an edge $e$ of $\Gs$}\\
ie^{-i\frac{\theta^\partial}{2}}
&\text{ if $\ws\bs$ is half a boundary dual edge of $\Gdext$, as in Figure \ref{fig:FigWeight4}},\\
ie^{i\frac{\theta^\partial}{2}}
&\text{ if $\ws\bs$ is half a boundary dual edge of $\Gdext$, as in Figure \ref{fig:FigWeight4}},\\
2\sin\frac{\theta^\partial}{2}&\text{ if $\ws\bs$ is half an edge of $\Gsext$ incident to $\rs$}.
\end{cases}
\end{equation}

\begin{figure}[ht]
\begin{center}
\psfrag{sin}[l][l]{\scriptsize $\sin\theta_e$}
\psfrag{cos}[l][l]{\scriptsize $i\cos\theta_e$}
\psfrag{te}[l][l]{\scriptsize $\theta_e$}
\psfrag{tb}[l][l]{\scriptsize $\theta^\partial$}
\psfrag{e}[l][l]{\scriptsize $e$}
\psfrag{l1}[l][l]{\scriptsize Weights $\tau$ around a boundary
white vertex}
\psfrag{l2}[l][l]{\scriptsize Weights $\tau$ around a non-boundary
white vertex}
\psfrag{l3}[l][l]{\scriptsize Primal and corresponding dual edge in the dual
configuration}
\psfrag{mb1}[l][l]{\scriptsize $\sin\theta^\partial$}
\psfrag{w1}[l][l]{\scriptsize $\frac{i}{2}e^{-i\frac{\theta^\partial}{2}}$}
\psfrag{w2}[l][l]{\scriptsize $\frac{i}{2}e^{i\frac{\theta^\partial}{2}}$}
\psfrag{G0}[l][l]{$\vec{\Gs}_0$}
\psfrag{G}[l][l]{$\vec{\Gs}$}
\includegraphics[width=10cm]{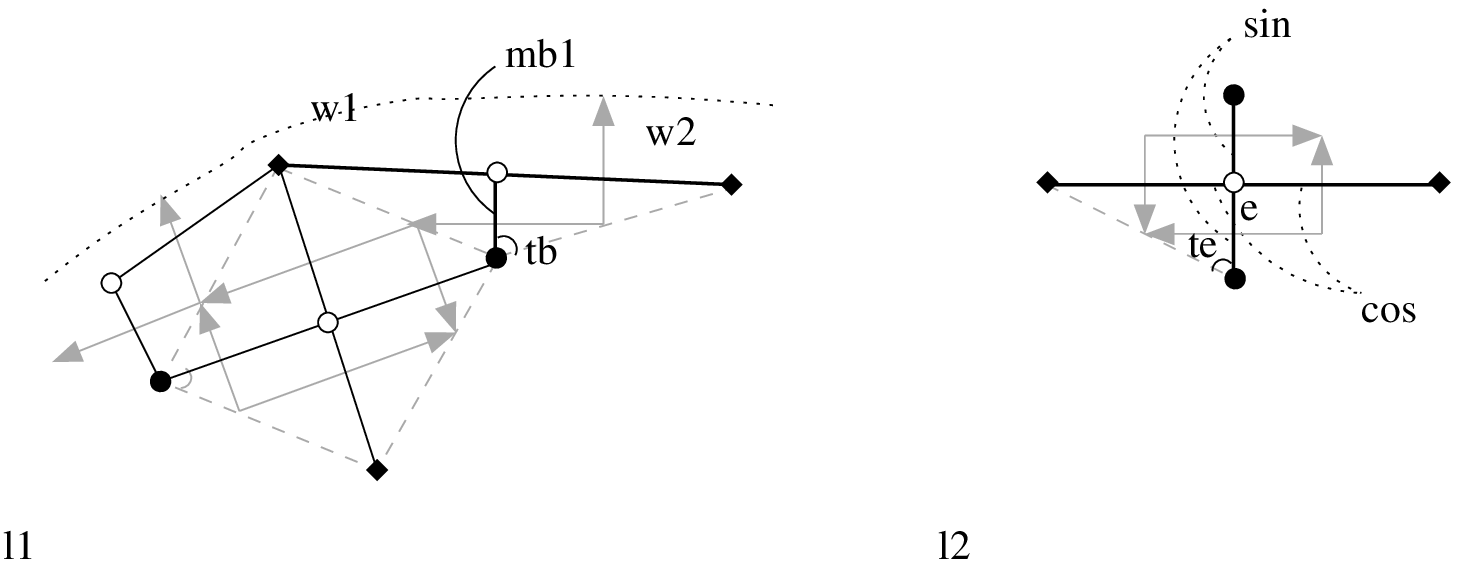}
\caption{Definition of the weights $\tau$ assigned to edges
of the extended double graph $\GDext$.}\label{fig:FigWeight4}
\end{center}
\end{figure}

Observing that non-boundary white vertices of $\GDext$ are in bijection with edges of $\Gs$, and that boundary white vertices
of $\GDext$ are in bijection with edges of $\Esext\setminus\Es$, \emph{i.e.} edges of $\Gsext$ incident to the boundary vertex $\rs$, we deduce
from Equations \eqref{equ:500}, \eqref{equ:501}, \eqref{equ:502}, \eqref{equ:503} and Lemma \ref{lem:union}, the following proposition.

\begin{prop}\label{prop:matchingtree}$\,$
\begin{enumerate}
\item Given a perfect matching of $\M(\GDext(\ss))$, the $\rho^*$-weighted sum $\rho^*(\T^{\ss}_{(1,2),\Ms}(\GDext))$ is equal to:
\begin{equation*}
 \rho^*(\T^{\ss}_{(1,2),\Ms}(\GDext))=
 \Bigl(\prod_{e\in\Es} ie^{-i\theta_e}\Bigr)
 \Bigl(\prod_{e\in\Esext\setminus\Es} -ie^{-i\frac{\theta^\partial}{2}}\Bigr)
 \prod_{\es\in\Ms} \tau_\es.
 \end{equation*}
\item The weighted sum of spanning trees of the extended double $\GDext$, satisfying Points 1. and 2. of Lemma \ref{lem:ConfigLoc} 
is equal, up to an explicit multiplicative constant, to the dimer partition function of the graph $\GDext(\ss)$, with weight function 
$\tau$:
$$
Z_{\mathrm{OST,(1,2)}}^{\ss}(\GDext,\rho^*)=\Bigl(\prod_{e\in\Es} ie^{-i\theta_e}\Bigr)
 \Bigl(\prod_{e\in\Esext\setminus\Es} -ie^{-i\frac{\theta^\partial}{2}}\Bigr)
 \Zdimer(\GDext(\ss),\tau).
$$
\end{enumerate}
\end{prop}

\subsection{Concluding the proof of Theorem \ref{thm:MainResult}}

From Theorem \ref{thm:CedBea} used with critical coupling constants, we have:
\begin{equation*}
(\Zising(\Gs,\Js))^2=2^{|\Vs|}\Bigl(\prod_{e\in\Es}\cos^{-1}\theta_e\Bigr)\Zdimer(\GQ,\nu).
\end{equation*}

From Corollary \ref{cor:DimerTrees1}, we have:
\begin{equation*}
\Zdimer(\GQ,\nu)=|\ZOST^{\rs}(\GO,\rho)|.
\end{equation*}

From Proposition \ref{prop:TreesDouble}, we have:
\begin{equation*}
\ZOST^{\rs}(\GO,\rho)=Z_{\mathrm{OST,(1,2)}}^{\ss}(\GDext,\rho^*).
\end{equation*}

Combining this with Point 2 of Proposition \ref{prop:matchingtree}, yields:
\begin{equation*}
(\Zising(\Gs,\Js))^2=2^{|\Vs|}\Bigl(\prod_{e\in\Es}\cos^{-1}\theta_e\Bigr)|\Zdimer(\GDext(\ss),\tau)|.
\end{equation*}

Then, by the work of Kenyon, Propp and Wilson \cite{KPW},
building on Temperley's bijection, we know that every perfect matching $\Ms$ of
$\M(\GDext(\ss))$ defines a pair of primal and dual oriented spanning trees of $\Gsext$ and
$\Gdext$, where the primal oriented spanning tree is rooted at the vertex $\rs$, and the dual one is rooted 
at the vertex $\ss$. In the bijection, half-edges of $\GDext$ correspond to oriented edges of the graphs $\Gsext$ and $\Gdext$,
and weights of oriented edges are obtained using the bijection. Starting from the weights of Equation \eqref{equ:weighttau2},
we nearly have the weights of Equation \eqref{equ:weighttau}. We need two more observations.

Since edges of $\Gs$ and $\Gs^*$ are unoriented, for every spanning tree $\Ts$ of $\Gsext$ and dual spanning tree $\Ts^*$
of $\Gdext$, we can write:
$$
\Bigl(\prod_{e\in\Es}\cos^{-1}\theta_e\Bigr)=
\Bigl(\prod_{e=(x,y)\in\Ts\cap\Es}\cos^{-1}\theta_e \Bigr)
\Bigl(\prod_{e^*=(u,v)\in\Ts^*\cap\Es^*}\cos^{-1}\theta_e \Bigr).
$$
This amounts to multiplying edge-weights $\tau$ of the graph $\Gs$ and of the graph $\Gs^*$ by $\cos^{-1}\theta_e$, yielding the 1 and 
$\tan\theta_e$ of \eqref{equ:weighttau} rather than $\cos\theta_e$ and $\sin\theta_e$ of \eqref{equ:weighttau2}.

Note that edge-weights of the extended dual graph $\Gdext$ differ by a constant $i$. This is because a spanning tree of a graph
contains a fixed number of edges, so that the factor $i$ on dual edges factors out, and does not change the modulus of the
partition function.

\bibliographystyle{alpha}
\bibliography{survey}

\end{document}